\documentclass[10pt,twocolumn,a4paper]{IEEEtran}

\usepackage[active]{srcltx} 
\usepackage[cmex10]{amsmath}
\usepackage{amsfonts}
\usepackage[final]{graphics}
\usepackage[final]{graphicx}
\usepackage{amsbsy}
\usepackage{amsbsy}
\usepackage{amssymb}
\usepackage{url}
\usepackage{enumerate}
\usepackage{cite}
\usepackage{psfrag}
\usepackage{color}
\usepackage{euscript}
\usepackage[utf8]{inputenc}
\usepackage{algorithmic}
\usepackage[ruled,vlined]{algorithm2e}

\newcommand{\bm}[1]{{\mathbf{#1}}}
\newcommand{\Es}{{\mathbb{E}}}          

\newcommand{\diag}{{\text{diag}}}

\newcommand{\I}{\bm{I}}
\newcommand{\Zero}{\bm{O}}

\newcommand{\ba}{\bm a}
\newcommand{\bg}{\bm g}

\newcommand{\br}{\bm r}

\newcommand{\bh}{\bm h}
\newcommand{\fb}{\bm f}

\newcommand{\pb}{\bm p}
\newcommand{\bs}{\bm s}

\newcommand{\zb}{\bm z}

\newcommand{\Cb}{\bm C}

\newcommand{\gb}{\bm g}

\newcommand{\Gammab}{\bm \Gamma}
\newcommand{\Thetab}{\bm \Theta}

\newcommand{\betab}{\boldsymbol \beta}
\newcommand{\gammab}{\boldsymbol \gamma}
\newcommand{\thetab}{\boldsymbol \theta}

\newcommand{\Fb}{\bm F}

\newcommand{\Rb}{\bm R}
\newcommand{\vb}{\bm v}

\newcommand{\xb}{\bm x}

\newcommand{\Bb}{\bm B}

\newcommand{\Cset}{\mathbb{C}}
\newcommand{\Rset}{\mathbb{R}}

\newcommand{\eqdef}{\triangleq}

\newcommand{\herm}{\text{H}}
\newcommand{\trasp}{\text{T}}

\newcommand{\pot}{\EuScript{P}}
\newcommand{\capa}{\EuScript{R}}
\newcommand{\capaa}{\EuScript{C}}
\newcommand{\fair}{\EuScript{F}}

\def\bdm#1\edm{\begin{displaymath}#1\end{displaymath}}
\def\be#1\ee{\begin{equation}#1\end{equation}}
\def\barr#1\earr{\begin{align}#1\end{align}}

\newcommand{\IeeeTIT}{{\em IEEE Trans.\ Inf. Theory\/}}
\newcommand{\IeeeTSP}{{\em IEEE Trans.\ Signal Process.\/}}
\newcommand{\IeeeTCOMM}{{\em IEEE Trans.\ Commun.\/}}

\newcommand{\IeeeWCOMMLETT}{{\em IEEE Wireless Commun.\ Lett.\/}}
\newcommand{\IeeeTWC}{{\em IEEE Trans.\ Wireless Commun.\/}}
\newcommand{\IeeeJSAC}{{\em IEEE J.\ Select.\ Areas Commun.\/}}
\newcommand{\IeeeTVT}{{\em IEEE Trans.\ Veh. Technol.\/}}
\newcommand{\IeeeTAP}{{\em IEEE Trans.\ Antennas Propag.\/}}

\newcommand{\IeeeCOMMMAG}{{\em IEEE Commun.\ Magazine\/}}

\newcommand{\IeeeACCESS}{{\em IEEE Access}}
\newcommand{\IeeeOJCOMMSOC}{{\em IEEE Open J.\ Commun.\ Soc.\/}}

\newtheorem{theorem}{Theorem}[section]

\newtheorem{lemma}[theorem]{Lemma}

\begin{document}

\title{Rapidly time-varying reconfigurable intelligent surfaces for 
downlink multiuser transmissions
}
\author{Francesco~Verde,~\IEEEmembership{Senior Member,~IEEE},
Donatella~Darsena,~\IEEEmembership{Senior Member,~IEEE},
and Vincenzo~Galdi,~\IEEEmembership{Fellow,~IEEE}

\thanks{
Manuscript received April 23, 2023; 
revised September 28, 2023 and 
December 1, 2023;
accepted January 21, 2024.
The associate editor coordinating the review of this paper and
approving  it for publication  was Dr.~Wenyi Zhang.
(\em Corresponding author: Francesco Verde)
}
\thanks{
F.~Verde and D.~Darsena are with the Department of Electrical Engineering and Information Technology,  University Federico II, Naples I-80125,
Italy, and also with the National Inter-University Consortium 
for Telecommunications (CNIT), Pisa I-56124, Italy 
[e-mail: (f.verde, darsena)@unina.it].
}
\thanks{
V.~Galdi is with the Department of Engineering, University of Sannio,  
Benevento I-82100, Italy (e-mail: vgaldi@unisannio.it).
}
\thanks{
This work was partially supported by the European Union under the Italian National Recovery and Resilience Plan (NRRP) of NextGenerationEU, partnership on ``Telecommunications of the Future" (PE00000001 - program ``RESTART").}
}
\markboth{IEEE Transactions on Communications, Vol.~xx,
No.~yy,~zz~2024}{Verde \MakeLowercase{\textit{et al.}}:
Rapidly time-varying reconfigurable intelligent surfaces for 
downlink multiuser transmissions}

\IEEEpubid{0000--0000/00\$00.00~\copyright~2024 IEEE}

\maketitle

\begin{abstract}
Amplitude, phase, frequency, and polarization
states of electromagnetic (EM) waves can be dynami\-cally manipulated 
by means of artificially engineered planar materials, 
composed of sub-wavelength meta-atoms, which are typically
referred to as metasurfaces.  
In wireless communications, metasurfaces 
can be configured to judiciously modify the EM propagation
environment in such a way to achieve different network objectives
in a digitally programmable way.  
In such a context, metasurfaces are typically referred to as 
reconfigurable intelligent surfaces (RISs).
Until now, researchers in wireless communications have mainly focused their attention 
on slowly time-varying designs of RISs, where 
the spatial-phase gradient across the RIS is varied at the rate 
equal to the inverse of the channel coherence time.
Additional degrees of freedom for controlling EM waves 
can be gained by applying a time modulation to the
reflection response of RISs during the channel
coherence time interval, thereby attaining
rapidly time-varying RISs. 
In this paper, we develop a general framework where 
a downlink multiuser transmission over  single-input 
single-output slow fading channels is assisted by a digitally controlled
rapidly time-varying RIS.
We show that reconfiguring the RIS at a rate greater than
the inverse of the channel coherence time  might be beneficial from a 
communication perspective depending 
on the considered network utility function and the available 
channel state information at the transmitter (CSIT). 
The conclusions of our analysis in terms of system design guidelines are as follows:
(i) if the network utility function is the  sum-rate 
time-averaged network capacity, without any constraint on fair resource allocation, and full CSIT is available, it is unnecessary to change the electronic properties of the RIS
within the channel coherence time interval;  
(ii) if partial CSIT is assumed only,  
a rapidly time-varying randomized RIS
allows to achieve 
a suitable balance between 
sum-rate time-averaged
capacity and user fairness, especially for a sufficiently large number of users;
(iii) regardless of the available amount of CSIT, the design of  
rapid temporal variations across the RIS is instrumental for developing 
scheduling algorithms aimed at maximizing the network capacity 
subject to fairness constraints.

\end{abstract}

\begin{IEEEkeywords}
Capacity, downlink transmission, fair scheduling, multiuser diversity, 
partial channel state information, reconfigurable intelligent surface (RIS), 
rapidly time-varying RIS, space-time metasurfaces.
\end{IEEEkeywords}

\section{Introduction}

\IEEEPARstart{S}{ince} the birth of radio communications in the early $1900$s,
engineers, mathematicians, and physicists have dealt with
the uncontrollable nature of propagation environments. 
{\em Reconfigurable intelligent surfaces (RISs)} are
emerging as an innovative technology for the electromagnetic
(EM) wave manipulation, signal modulation, and
smart radio environment reconfiguration in 
wireless communications systems \cite{Liaskos.2018, Basar.2019,
DiRenzo.2019, DiRenzo.2020, Tang.2020, Tang.2021}.
An RIS is a metasurface composed of 
uniform or non-uniform sub-wavelength metal/dielectric structures, referred to as 
{\em meta-atoms} (or {\em unit cells}),  
which interact with the incident EM fields, 
whose working frequency can vary from sub-$6$ GHz to THz.
Such meta-atoms can be independently controlled via software 
to switch among different reflection amplitude and phase 
responses that collectively shape the
wavefront of the impinging wave \cite{Cui.2014,Huang.2017},
without the need of bulky phase shifters.
Thanks to their ultra-thin profile, square-law array
gain, low-power consumption, low noise, and easy fabrication and deployment, RISs 
yield more freedom and flexibility to create different reflecting
patterns for incident waves from diverse directions \cite{Li.2016, Li.2017, Cai.2019, Liu.2017, Sievenpiper.2003}, compared to conventional phased antenna arrays.
RISs can be realized by embedding  into 
meta-atoms  tunable varactor diodes, switchable
positive-intrinsic-negative (PIN) diodes, 
field effect transistors, or
micro-electromechanical system switches, 
which collectively
form an inhomogeneous surface with specific phase
profiles \cite{Lau.2012,Boccia.2012,Shad.2012,Xu.2013,
Chen.2017}.

\IEEEpubidadjcol

In wireless communications,  
the channel response may change with time. 
The channel coherence time is a measure of the expected
time duration over which the response of the channel is essentially invariant
and it is related to the bandwidth of the Doppler spectrum.
Due to the time-varying nature of wireless channels, 
the reflection coefficient of the meta-atoms 
has to be varied in both space and time domains,
thus leading to the so-called {\em space-time metasurfaces}
in the EM literature \cite{Zhang.2020,Zhao.2019,Zhang.2018,Tennant.2009,Kummer.1963,Zhang.2019,Dai.2018,Dai.2020a,Dai.2020b,Raj.2021,Castaldi_2021,Zhang_2021,Tara_2022, 
Zhao_2019,Wang_2021,Hada_2015,Taravati_2020}.
In {\em slowly time-varying} design techniques, 
the rate at which the reflecting coefficients 
of the meta-atoms are updated is equal to the inverse
of the channel coherence time.
For each channel coherence time interval,  
anomalous reflection/refraction (for beam steering) 
and reflective/transmissive focusing (for high directivity)
are obtained by spatially varying the 
state of meta-atoms at different 
positions on the RIS.  
In this paper, we instead consider 
{\em rapidly time-varying} design techniques: 
in this case, besides constructing the
spatially varying phase profile, 
the phase responses of the meta-atoms are modulated  
in time  during the channel coherence time interval, 
by assuming that the time-modulation speed is much smaller than the
frequency of the incident EM signal.
Rapidly time-varying RISs offer additional degrees of freedom for 
manipulating signals, such as, e.g., nonlinear harmonic manipulations \cite{Zhao_2019}
and spread-spectrum camouflaging \cite{Wang_2021}.

Relying on the feasibility of engineering 
the meta-atoms on the metasurface, a plethora of papers has been published 
where the wireless channel is programmed 
by optimizing on-the-fly the response  of built-in RISs to 
achieve different system objectives, such as 
capacity gains \cite{Hu.2018, Jung.2020, Li.2020,Di.2020,Huang.2020,Najafi.2021,Abrardo.2021},
coverage extension \cite{Wang.2020,Pan.2020b}, 
energy efficiency \cite{Huang.2019,Yan.2020},
maximization of minimum signal-to-noise ratio (SNR) \cite{Nadeem.2020},
joint maximization of rate and energy efficiency \cite{Zappone.2021},
quality of services (QoS) targets under imperfect channel state information (CSI) \cite{Zhou.2020},
super resolution CSI estimation \cite{He.2021}, and 
high level of localization accuracy \cite{Dardari.2022}.
When the control of the RIS is 
performed by the transmitter, the RIS is also equipped 
with a smart controller that communicates with the
transmitter via a separate wireless link for coordinating transmission
and exchanging information on the phase shifts of all 
reflecting elements in real time. 
In this scenario, {\em CSI at the
transmitter (CSIT)} is required to derive the optimal 
response of the RIS.

To the best of our knowledge, {\em almost all the studies
within the wireless communications community consider only  
slowly time-varying RISs}.\footnote{The list 
of the aforementioned references is by no means exhaustive. We refer to the 
numerous surveys available in the literature for a clearer
picture of the state-of-the-art regarding slowly time-varying designs 
of RISs in wireless communications systems.
}
Notable exceptions are represented by \cite{Yur_2020} and  \cite{Mizmizi_2023}.
Along the same lines of \cite{Zhang.2018},  only
point-to-point communication (i.e., the scenario of a single transmitter and a single receiver) has
been considered in \cite{Yur_2020}.
In \cite{Mizmizi_2023}, the modulation in time of the reflection
response of the RIS is used to transmit information over a full-duplex 
point-to-point  link. 
Nevertheless, the use of rapidly time-varying 
RISs specifically subordinated to the interest 
of a network of many users (i.e., RISs do not transmit their own information)
remains an unexplored issue and their benefits in  
multiuser communications have not been studied yet.

\subsection{Contributions}
\label{sec:contributions}

We focus on a single-input single-output downlink of a wireless communication 
multiuser system operating over a slow fading channel with coherence time $T_\text{c}$,
which is assumed to be much longer than the symbol period of the transmit digital modulator. 
We assume a delay-tolerant scenario, where the user rates can be adapted
according to their instantaneous channel conditions. 
The scheduler function in the 
medium-access-control sublayer of the transmitter consists of 
dividing each channel coherence time interval  in time slots
and introducing a scheme for resource assignment among the users. 
The downlink transmission is assisted by a digital rapidly time-varying RIS, whose
reflection response is fixed during each time slot, but it 
can be digitally varied in both time and space domains from
one time slot to another in a controlled manner.  
Besides channel conditions, the total system performance and the data rate of each user are affected by the scheme used for scheduling and the states of the reflection coefficients
of the RIS meta-atoms.
The contributions of this study are summarized as follows.

\begin{enumerate}[1)]

\item
We prove that maximization of the sum-rate time-averaged capacity,
with respect to transmit powers allocated at the users and reflection coefficients of the meta-atoms, can
be achieved by a slowly time-varying RIS in the case of full CSIT. 
In this case, the reflection coefficients can be optimized with a manageable computational complexity
by resorting to 
the {\em block-coordinate descent} (or {\em Gauss-Seidel}) method \cite{Bert.1999},
which accounts for discrete variations  of the reflection response of the RIS.

\item
The scheduling deriving by the maximization of the sum-rate time-averaged capacity
with a slowly time-varying RIS would result in a very unfair sharing of the channel resource, by letting only the user with the best channel condition to transmit
during each channel coherence time interval. Moreover, in many applications, it is not 
reasonable to assume that full CSIT can be made available. 
To partially overcome such shortcomings, 
we propose to use a rapidly time-varying randomized RIS where the reflection coefficients of the
meta-atoms are randomly varied from one time slot to another in both space and time. 
This induces 
random fluctuations in the overall channel seen by each user
even if the physical channel has no fluctuations over a time interval of duration $T_\text{c}$.
In this case, 
maximization of the sum-rate time-averaged capacity with respect 
to transmit powers allocated at the users requires only partial
CSIT and more than one user may be scheduled during each channel coherence time interval.

\item
To ensure a proportionally fair allocation of channel resource \cite{Kelly.1998},
we show that, when the aim is to  maximize the sum of logarithmic 
exponentially weighted moving average user rates, a rapidly time-varying 
RIS is the natural solution in the case of both full and partial CSIT.

\end{enumerate}
 
\subsection{Notations}
\label{sec:pre}

Upper- and lower-case bold letters denote matrices and vectors;
the superscripts
$*$, $\trasp$ and $\herm$
denote the conjugate,
the transpose, and the Hermitian (conjugate transpose) of a matrix;
$[x]^{+}$ stands for $\max\{x,0\}$;
$\mathbb{C}$, $\mathbb{R}$ and $\mathbb{Z}$ are
the fields of complex, real and integer numbers;
$\mathbb{C}^{n}$ $[\mathbb{R}^{n}]$ denotes the
vector-space of all $n$-column vectors with complex
[real] coordinates;
similarly, $\mathbb{C}^{n \times m}$ $[\mathbb{R}^{n \times m}]$
denotes the vector-space of all the $n \times m$ matrices with
complex [real] elements;
$j \eqdef \sqrt{-1}$ denotes the imaginary unit;
$\Re\{x\}$ is the real part of $x \in \Cset$;
$\ast$ denotes the (linear) convolution operator;
$\Pi(t)$ is the basic rectangular window, i.e., $\Pi(t)=1$ for $|t|\le1/2$ and zero otherwise;
$\lfloor x \rfloor$ rounds $x \in \mathbb{R}$ to the nearest integer less than or equal to that element;
$\otimes$ is the Kronecker product;
$\mathbf{0}_{n}$, $\Zero_{n \times m}$ and $\I_{n}$
denote the $n$-column zero vector, the $n \times m$
zero matrix and the $n \times n$ identity matrix;
matrix $\mathbf{A}= \diag (a_{0}, a_{1}, \ldots,
a_{n-1}) \in \mathbb{C}^{n \times n}$ is diagonal;
$\{\mathbf{A}\}_{i,\ell}$ indicates the
$(i,\ell)$th element of $\mathbf{A} \in \Cset^{n \times m}$, with 
$i \in \{1,2,\ldots, n\}$ and $\ell \in \{1,2,\ldots, m\}$;
$\|\mathbf{a}\|$ denotes the 
Frobenius norm of $\mathbf{a} \in \Cset^{n}$;
$\Es[\cdot]$ denotes ensemble averaging; 
$\xb \sim {\cal CN}(\boldsymbol{\mu}, \Cb)$ is
a circularly symmetric complex
Gaussian vector $\xb \in \Cset^n$ with mean 
$\boldsymbol{\mu} \in \Cset^n$ and
covariance $\Cb \in \Cset^{n \times n}$.

\begin{figure}[t]
\centering
\includegraphics[width=\columnwidth]{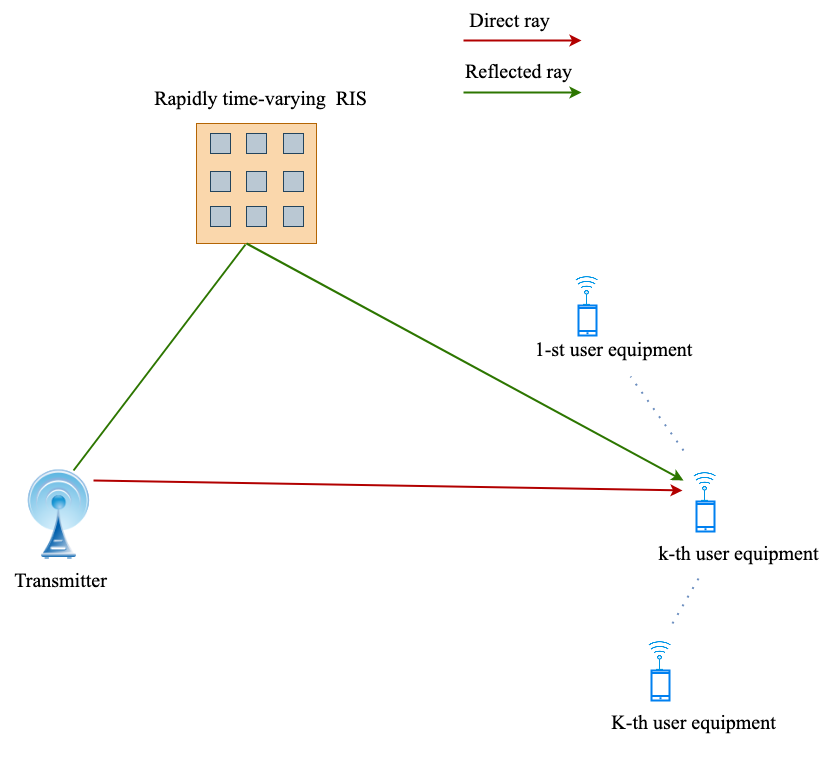}
\caption{A pictorial view of a downlink transmission in a 
multiuser system assisted by a rapidly time-varying RIS.
In the figure, the direct and reflected rays of the generic $k$-th user are depicted only,
with $k \in \{1,2,\ldots, K\}$.
}
\label{fig:fig_1}
\end{figure}

\section{Downlink with rapidly time-varying RIS}
\label{sec:system}

With reference to Fig.~\ref{fig:fig_1}, 
we consider a wireless system including a 
single-antenna transmitter, whose aim
is to transfer information in downlink to $K > 1$
single-antenna user equipments  (UEs).
The transmission is assisted by a digitally programmable 
RIS working in reflection mode, which is 
installed on the face of a surrounding  high-rise building, 
indoor walls, or roadside billboards.
The RIS is made of $Q=Q_x \times Q_y$ spatial meta-atoms,
which are positioned along a rectangular grid 
having $Q_x$ and $Q_y$ elements on the $x$ and $y$ axes, respectively. 
The meta-atoms are independently controlled to switch
among different EM responses that collectively shape the
spatial and spectral characteristics 
of the illuminating wave 
in a time-varying fashion \cite{Zhang.2020,Zhao.2019,Zhang.2018,Tennant.2009,Kummer.1963,Zhang.2019,Dai.2018,Dai.2020a,Dai.2020b,Raj.2021,Castaldi_2021,Zhang_2021}.
The functionality of such a
metasurface is determined by the instantaneous digital reflection response of
each meta-atom, which can be dynamically controlled by
digital logic devices \cite{Cui.2014,Huang.2017}.
We consider a transmission system with baud rate $1/T_\text{s}$.
Throughout the paper,  we assume
that the communication bandwidth $W \propto 1/T_\text{s}$
is much smaller than the carrier frequency $f_0$ ({\em narrowband assumption}),
i.e., $f_0 \, T_\text{s} \gg 1$.
Narrowband frequency-flat channels\footnote{Wideband 
channels can be turned into flat fading ones
by multicarrier techniques that are  
widely applied in broadband wireless standards \cite{LTE,5G}.
}
are considered with coherence time $T_\text{c}$
spanning $L_\text{c}$ symbol 
intervals, i.e., $T_\text{c}=L_\text{c} \, T_\text{s}$. 
These channels are modeled by a single-tap filter, as most of the multipaths arrive
during one symbol time,  which is approximately
constant within frame intervals of length $T_\text{c}$,
but is allowed to independently change from one
coherence period to another ({\em block fading model}).
Each channel tap coefficient is formed by the superposition of a large
number of micro-scattering components (e.g., due to rough
surfaces) having (approximately) the same delay. By virtue of the central 
limit theorem, it is customary to model
the superposition of such  many small effects as Gaussian \cite{Proakis}.

The signal received by each user consists of two components (see Fig.~\ref{fig:fig_1}): 
the {\em direct} ray, which corresponds to
the signal propagating from the transmitter to the $k$-th UE
with propagation delay $\tau_{0,k}$, 
and a {\em reflected} ray, which is the 
signal reflected off the RIS with (overall) propagation
delay $\tau_{1,k}$.
The delay spread $\Delta_k \eqdef \tau_{1,k}-\tau_{0,k}$
of the {\em composite} $k$-th channel (i.e., direct and reflected rays)
equals the difference between the delay of the direct ray and that 
of the reflected one.
We assume that the transmitted signal is also narrowband 
relative to the delay spread,\footnote{This assumption can be relaxed by 
considered a block transmission system that
introduces a time guard band 
of duration $\Delta_\text{max} \eqdef \max_{k \in \{1,2,\ldots,K\}} \Delta_k$
after every block to
eliminate the interblock interference between data blocks.} i.e., 
$\Delta_k \ll T_\text{s}$, such that
the effects of $\tau_{0,k}$ and $\tau_{1,k}$ can be ignored at the $k$-th UE.

\subsection{Signal emitted by the transmitter}
\label{sec:DPC}

Each channel coherence time interval is divided into 
$M$ time slots of duration $T_\text{r} \eqdef T_\text{c}/M$.
Each slot spans $P$ symbols intervals, i.e., $T_\text{r} 
= P \, T_\text{s}$. Consequently, it results that $L_\text{c} = M \, P$.
Let $s_k(n)$ be the information-bearing symbol intended for the
$k$-th user within the $n$-th symbol interval, for 
$k \in \{1,2,\ldots, K\}$ and $n \in \mathbb{Z}$.
The transmitted symbols $\{s_k(n)\}_{n \in \mathbb{Z}}$ are modeled as
mutually independent sequences of zero-mean unit-variance independent and identically
distributed (i.i.d.) complex circular random variables.
For each $k \in \{1,2,\ldots, K\}$, the data stream $\{s_k(n)\}_{n \in \mathbb{Z}}$ 
of the $k$-th user at rate $1/T_\text{s}$ is converted
at the transmitter side into $M$ parallel substreams
$s^{(m)}_k(p) \eqdef s_k(p \, P + m)$, for 
$p \in \{0,1,\ldots, P-1\}$ and 
$m \in \{0,1,\ldots, M-1\}$.
Such substreams are superposed as follows
\be
s^{(m)}(p) = \sum_{k=1}^K \sqrt{\pot_k^{(m)}} \, s^{(m)}_k(p)
\ee
where $\pot_k^{(m)}$ is the transmit power allocated to the 
$k$-th user in the $m$-th slot, subject to the time-averaged sum-power constraint
$\frac{1}{M} \sum_{k=1}^K \sum_{m=0}^{M-1}
\pot_k^{(m)} \le \pot_\text{TX}$,
where $\pot_\text{TX} >0$ denotes the (fixed) 
maximum allowed transmit power.
With reference to a generic coherence interval, 
the continuous-time baseband signal emitted by the transmitter reads as
\be
x(t)=\sum_{m=0}^{M-1} \sum_{p=0}^{P-1}  s^{(m)}(p) 
\, \psi\left(t-p \, T_\text{s}- m \, T_\text{r}\right)
\label{eq:xk}
\ee
where $\psi(t)$ is the unit-energy square-root Nyquist pulse-shaping
filter \cite{Proakis}. In this case, the communication bandwidth 
is $W=(1+\varrho_\text{s})/T_\text{s}$,
with $0 \le \varrho_\text{s} \le 1$ being the rolloff factor.

\subsection{Rapidly time-varying RIS}
\label{sec:signal-R}

The RISs widely studied in existing works 
\cite{Cui.2014,Huang.2017,Li.2016, Li.2017, Cai.2019, Liu.2017, Sievenpiper.2003, 
Lau.2012,Boccia.2012,Shad.2012,Xu.2013, Chen.2017,
Zhang.2020,Zhao.2019,Zhang.2018,Tennant.2009,Kummer.1963,Zhang.2019,Dai.2018,Dai.2020a,Dai.2020b,Raj.2021,Castaldi_2021,Zhang_2021}
consume minimal power, which can be 
as low as a few $\mu$W per meta-atom,
and the introduced thermal noise is also negligible. 
By assuming that 
the reflection-response of the RIS is 
frequency flat, i.e., it is constant
over the whole signal bandwidth $W$,
the far-field baseband signal reflected by a rapidly time-varying RIS reads as
\be
\zb(t) =\Rb(t) \, \gb \, x(t)
\label{eq:z-vet}
\ee
with $\Rb(t) \eqdef \diag\left[\gamma_{1}(t), 
\gamma_{2}(t), \ldots, \gamma_{Q}(t)\right]$, where
according to the time-switched array theory \cite{Kummer.1963, Tennant.2009},
the complex parameter
$\gamma_{q}(t)$ is the {\em time-modulated reflection
coefficient} of the $q$-th meta-atom,
$\forall q \in \{1,2, \ldots, Q\}$, whereas $\gb \in \Cset^{Q}$ 
denotes  the low-pass equivalent channel response
between the transmitter and the RIS
within the considered channel coherence time.
RISs have the capability of controlling the amplitude,
phase, and polarization of incident EM waves \cite{Ding.2020}.
Design methods of programmable RISs
have been proposed in   
\cite{Liao.2021,Hong.2021,Li.2022} where the reflection magnitude
and phase can be independently controlled. 
We focus on an RIS with 
phase-only modulation, by deferring the 
case of RISs with amplitude-phase modulated elements 
to future work.

If $\Rb(t)$ does not vary 
during each channel coherence time interval,  
the effect of the RIS consists of
multiplying the impinging signal $\gb \, x(t)$
by a constant matrix $\Rb(t) \equiv \Rb$, i.e., slowly time-varying RIS, 
which is the configuration considered in previous works 
\cite{Liaskos.2018, Basar.2019,
DiRenzo.2019, DiRenzo.2020, Tang.2020, Tang.2021,
Hu.2018, Jung.2020, Li.2020,Di.2020,Huang.2020,Najafi.2021,Abrardo.2021,
Wang.2020,Pan.2020b,Huang.2019,Yan.2020, Nadeem.2020,
Zappone.2021,Zhou.2020,He.2021,Dardari.2022}: in this case, each UE sees 
a time-invariant overall channel (see also Subsection~\ref{sec:TI}).
In  this  paper,  relying on the results of recent experimental
studies \cite{Zhang.2020,Zhao.2019,Zhang.2018,Tennant.2009,Kummer.1963,Zhang.2019,Dai.2018,Dai.2020a,Dai.2020b,Raj.2021,Castaldi_2021,Zhang_2021,Hada_2015,Taravati_2020}, 
we  consider  the  case  where $\Rb(t)$ is
time-varying over the channel 
coherence time interval.
In its simplest form, the reflection coefficient of the generic $q$-th meta-atom 
is subject to a linear modulation as
\be
\gamma_{q}(t)  = 
\sum_{m=0}^{M-1} \gamma_q^{(m)} \, \phi(t- m \, T_\text{r})
\label{eq:gamma}
\ee
where 
$\phi(t)=\Pi\left(\frac{t-T_\text{r}/2}{T_\text{r}}\right)$
is a pulse-shaping filter of duration $T_\text{r}$,\footnote{The function 
$\phi(t)$ can be of a some other appropriate shape 
\cite{Zhang.2019,Castaldi_2021,Liao.2021,Hong.2021,Li.2022}
and it represents another degree of freedom in designing the RIS
that we aim at exploring  in a future work.
} 
the complex coefficients 
$\gamma_q^{(0)}, \gamma_q^{(1)}, \ldots, \gamma_q^{(M-1)}$
are at the designer's disposal. For an $b$-bit
digital meta-atom with phase-only modulation capability, 
each coefficient $\gamma_q^{(m)}$ assumes values in the set  
$\mathcal{R} \eqdef \{R_0, R_1, \ldots, R_{2^b-1}\}$
of cardinality $2^b$, 
with elements
$R_\ell \eqdef \exp(j \frac{2 \pi}{2^b} \ell)$, for $\ell \in \{0,1, \ldots, 2^b-1\}$.
By considering the balance between the system cost,
design complexity, and overall performance, it has been found
that $b$-bit meta-atoms are suitable for RIS implementation \cite{Wu.2008},
for small-to-moderate values of $b$.

\begin{table*}
\scriptsize
\centering
\caption{Summary of the main system parameters 
and the time-scale of change of the key physical quantities.}
\label{tab:example-1}
\begin{tabular}{ccc}
\hline
\noalign{\vskip\doublerulesep}
\textbf{Key parameters and time-scales} & \textbf{Symbols} & 
\textbf{Values for an exemplifying scenario} 
\tabularnewline[\doublerulesep]
\hline
Carrier frequency & $f_0$ & $1.5$ GHz 
\tabularnewline
Communication  bandwidth & $W$ & $5$ MHz
\tabularnewline
Rolloff factor  &  $\varrho_\text{s}$ & $0.25$
\tabularnewline
Symbol period & $T_\text{s}=(1+\varrho_\text{s})/W$ & $250$ ns
\tabularnewline
Channel coherence time  &  $T_\text{c}$ & $50$ ms
\tabularnewline
Number of symbols per channel coherence interval & 
$L_\text{c}=\lfloor T_\text{c}/T_\text{s} \rfloor$  & $200,000$ 
\tabularnewline
Switching frequencies of meta-atoms & $f_\text{r}=1/T_\text{r}$ & $50$ kHz
\tabularnewline
Number of slots per channel coherence interval & 
$M=\lfloor T_\text{c}/T_\text{r} \rfloor$  & $2500$ 
\tabularnewline
Number of symbols per slot & 
$P=\lfloor T_\text{r}/T_\text{s} \rfloor$  & $80$ 
\tabularnewline
\hline
\end{tabular}
\end{table*}

According to \cite{Zhang.2020,Zhao.2019,Zhang.2018,Tennant.2009,Kummer.1963,Zhang.2019,Dai.2018,Dai.2020a,Dai.2020b,Raj.2021,Castaldi_2021,Zhang_2021,Hada_2015,Taravati_2020},
eq.~\eqref{eq:z-vet} is approximately
valid when the reflection coefficients 
vary at a very slow pace compared to the carrier
frequency of the impinging signal, i.e., 
the {\em switching frequency} $f_\text{r} \eqdef 1/T_\text{r}$ is much smaller
than $f_0$ or, equivalently, $f_0 \, T_\text{r} \gg 1$. This operative condition 
is surely satisfied in the narrowband regime since $T_\text{r} > T_\text{s}$ by assumption.
The main system parameters and the time-scale of change of key quantities are summarized in the first two columns of Table~\ref{tab:example-1}, 
whereas their values are reported in the third column for 
an exemplifying scenario of practical interest.
In particular, the values of the carrier frequency $f_0$ and the 
communication  bandwidth $W$ are consistent with 
Frequency Range 1 (FR1) of 5G New Radio (NR) \cite{3GPP-FR}.
The value of the channel coherence time is typical of a pedestrian scenario.
In current implementations \cite{Zhang.2020,Zhao.2019,Zhang.2018,Tennant.2009,
Kummer.1963,Zhang.2019,Dai.2018,Dai.2020a,Dai.2020b,Raj.2021,Castaldi_2021,Zhang_2021,
Taravati_2020}, the switching frequencies of the meta-atoms can reach
tens of  MHz. Larger values of $f_\text{r}$ can be achieved 
by exploiting faster switching mechanisms based on 
innovative materials, such as graphene or 
vanadium dioxide \cite{Wang.2019}.
However, the time variation induced by the RIS 
should be slow enough and should happen at a time scale allowing 
the reflection response $\Rb(t)$ to be unaffected by the 
receiving filter at the UEs (see Subsection~\ref{sec:Fourier}). 
Further, the variation should be slow 
enough to ensure that CSI is reliably acquired (see Subsection~\ref{sec:channels}).  

\subsection{Signal received by each user}
\label{sec:Fourier}

We customarily assume that each UE uses standard timing synchronization with respect
to its direct link. To decode the data signal, each
UE performs matched filtering with respect to the symbol
pulse $\psi(t)$. 
After matched filtering and sampling with rate $1/T_\text{s}$ at time instants 
$t_{p}^{(m)} \eqdef p \, T_\text{s} + m \, T_\text{r}$,  with $p \in \{0,1, \ldots, P-1\}$
and $m \in \{0,1, \ldots, M-1\}$, the discrete-time baseband 
signal received at the $k$-th user reads as
\begin{multline}
r_{k}^{(m)}(p)  = \underbrace{h_{k}^* \, x(t) \ast \psi(-t)}_\text{Direct link} \Big|_{t=t_{p}^{(m)}} \\ + 
\underbrace{\fb_{k}^\herm \, \zb(t) \ast \psi(-t)}_\text{Reflection link} \Big|_{t=t_{p}^{(m)}} 
+ v_{k}^{(m)}(p)
\label{eq:sig}
\end{multline}
for $k \in \{1,2,\ldots, K\}$, 
where $h_{k} \in \Cset$ models the low-pass equivalent channel 
response from the transmitter to user $k$, the vector $\fb_{k}  \in \Cset^Q$ denotes  the low-pass equivalent channel 
response from the RIS to the $k$-th user, 
and $v_{k}^{(m)}(p) \sim {\cal CN}(0,1)$ is the noise sample
at the output of the matched filter, with 
$v_{k_1}^{(m_1)}(p_1)$ statistically
independent of $v_{k_2}^{(m_2)}(p_2)$, for 
$k_1 \neq k_2$, $p_1 \neq p_2$, and $m_1 \neq m_2$.
Without loss of generality, we have assumed that the noise
has unit variance. Hence, the parameter $\pot_\text{TX}$ in Subsection~\ref{sec:DPC}
takes on the meaning of total {\em transmit} signal-to-noise ratio (SNR).

By substituting \eqref{eq:xk}-\eqref{eq:z-vet} in \eqref{eq:sig}, 
and accounting for \eqref{eq:gamma}, the block 
$\br_k^{(m)} \eqdef [r_{k}^{(m)}(0), r_{k}^{(m)}(1), \ldots, r_{k}^{(m)}(P-1)]^\trasp \in \Cset^P$ received 
by the $k$-th user
within the $m$-th slot is given by\footnote{Since the raised cosine pulse 
$\psi(t) \ast \psi^*(-t)$ satisfies the Nyquist criterion, its effect
disappears in the discrete-time model \cite{Proakis}.}
\be
\br_k^{(m)}   =  
[c_{k}^{(m)}]^* \left[ \sum_{u=1}^K  \, \sqrt{\pot_u^{(m)}}\, \bs_{u}^{(m)}\right] +  \vb_{k}^{(m)}
\label{eq:rkt}
\ee
where the {\em overall} channel gain seen by the $k$-th UE during the $m$-th slot can be written as  
\be
c_{k}^{(m)} \eqdef h_{k} + \gb^\herm \, [\Gammab^{(m)}]^* 
\, \fb_{k} \in \Cset\: , \quad \text{for $k \in \{1,2,\ldots, K\}$}
\label{eq:vet-ck}
\ee 
with
\barr
\Gammab^{(m)} & \eqdef \diag\left[\gamma_1^{(m)},
\gamma_2^{(m)}, \ldots, \gamma_Q^{(m)}\right]
\\
\bs_k^{(m)} &
\eqdef [s_{k}^{(m)}(0), s_{k}^{(m)}(1), \ldots, s_{k}^{(m)}(P-1)]^\trasp \in \Cset^P
\\
\vb_k^{(m)} & \eqdef [v_{k}^{(m)}(0), v_{k}^{(m)}(1), \ldots, v_{k}^{(m)}(P-1)]^\trasp \in \Cset^P \: .
\earr
In eq.~\eqref{eq:rkt}, we have used  
$[\Rb(t) \, \gb \, x(t)] \ast \phi(-t) \approx \Rb(t) \, \gb \, [x(t) \ast \phi(-t)]$, which 
holds if $\Rb(t)$ varies at a very slow pace compared to the effective
duration of $\phi(t)$, i.e., $f_\text{r} \ll 1/T_\text{s}$ or, equivalently, 
$P \gg 1$. In many applications, the channel coherence time is of the order of 
hundreds to thousands symbols, and therefore sufficiently large values of $P$ can be chosen in 
these operative scenarios (see Table~\ref{tab:example-1} for instance).

A key observation is now in order.
The received signal \eqref{eq:rkt} shows that, even though the channels
$h_{k}$, $\gb$, and $\fb_{k}$ are time invariant 
over each coherence interval spanning $L_\text{c}$ consecutive symbol periods,
the overall channel \eqref{eq:vet-ck}
seen by each user varies in time on a slot-by-slot basis.
The rate of variation of $c_{k}^{(m)}$ in time is dictated by
$M$, as well as the diagonal matrices $\Gammab^{(0)}, \Gammab^{(1)}, \ldots, 
\Gammab^{(M-1)}$, which can be suitably 
controlled to dynamically schedule resources among the users
and, thus, they are design parameters of the system. 
Strictly speaking, the overall effect of the rapidly time-varying RIS is 
to increase the rate and dynamic range of channel fluctuations.  

\subsection{The conventional case of a slowly time-varying RIS}
\label{sec:TI}

In many recent works \cite{Liaskos.2018, Basar.2019,
DiRenzo.2019, DiRenzo.2020, Tang.2020, Tang.2021,
Hu.2018, Jung.2020, Li.2020,Di.2020,Huang.2020,Najafi.2021,Abrardo.2021,
Wang.2020,Pan.2020b,Huang.2019,Yan.2020, Nadeem.2020,
Zappone.2021,Zhou.2020,He.2021,Dardari.2022}, the RIS 
does not introduce time variability during the transmission of an 
information-bearing symbol block, i.e., it is slowly time-varying.
This can be regarded as a particular case of our general model, where
$M=1$ and $T_\text{r}=T_\text{c}$ (see Fig.~\ref{fig:fig_2}).
In this special case, one has 
$\gamma_q(t) = \gamma_q^{(0)}$, 
for  $q \in \{1,2, \ldots, Q\}$.
Therefore, if the reflection process of the RIS is time invariant over each channel coherence time interval, 
that is, 
$\Gammab^{(0)}=\Gammab^{(1)}=\cdots=
\Gammab^{(M-1)}=\diag[\gamma_1^{(0)},
\gamma_2^{(0)}, \ldots, \gamma_Q^{(0)}]$,
the overall channel vector \eqref{eq:vet-ck} from the transmitter to each user 
becomes constant over each channel coherence time interval. 

\subsection{Assumptions regarding the channels}
\label{sec:channels}

It is common to model the channels between the transmitter and 
the meta-atoms of the RIS by assuming that the
received wavefront is planar (see, e.g., \cite{Dri.1999}). 
In this case, the transmitter-to-RIS channel vector $\gb$ can be modeled as    
\be
\gb = \sigma_g  \left(\sqrt{\frac{\kappa}{\kappa+1}} +
\sqrt{\frac{1}{\kappa+1}} \, g_\text{DIF}\right) \ba_\text{RIS} 
\label{eq:Rice}
\ee
where the first summand represents the
(deterministic) {\em line-of-sight (LoS)} component,
whereas the second one    
is the (random) {\em diffuse (i.e., non-LoS)} component,
with $g_\text{DIF} \sim {\cal CN}(0, 1)$. 
The variance $\sigma_g^2$
describes the large-scale geometric pathloss
and $\kappa \ge 0$ is the Ricean factor.
It is worthwhile to note that 
$\kappa=0$ yields  Rayleigh  fading, Rician fading is obtained when $ \kappa >0$,
and,  
in the limiting case $\kappa \to + \infty$, there is only 
the LoS channel component.
The {\em spatial signature} $\ba_\text{RIS}$ of the RIS in \eqref{eq:Rice}
is given by
\begin{multline}
\ba_\text{RIS} \eqdef 
\left[1, e^{j \frac{2 \pi}{\lambda_0} d_\text{RIS} u_x}, \ldots, e^{j \frac{2 \pi}{\lambda_0} (Q_x-1) d_\text{RIS} u_x}\right]^\trasp \\
\otimes  \left[1, e^{j  \frac{2 \pi}{\lambda_0} d_\text{RIS} u_y}, \ldots, e^{j  \frac{2 \pi}{\lambda_0} (Q_y-1) d_\text{RIS} u_y}\right]^\trasp \in \mathbb{C}^{Q}
\end{multline}
where   $\lambda_0=c/f_0$ is the wavelength, 
with $c$ being the speed of light in the medium, 
$d_\text{RIS}$ is the inter-element spacing
at the RIS, $\theta_\text{RIS} \in [0, 2 \pi)$ and 
$\phi_\text{RIS} \in [-\pi/2, \pi/2)$
identify the azimuth and elevation angles, respectively,
whose directional cosines are 
$u_x \eqdef \sin \theta_\text{RIS} \, \cos \phi_\text{RIS}$
and $u_y \eqdef \sin \theta_\text{RIS} \, \sin \phi_\text{RIS}$.
It should be noted that $\Es[\|\bg\|^2]=\sigma_g^2 \, Q$.

The random channels $h_k$ and $\fb_k$ are modeled as 
$h_k \sim {\cal CN}(0, \sigma_{h_k}^2)$ and 
$\fb_k \sim {\cal CN}(\bm{0}_Q, \sigma_{f_k}^2\, \I_{Q})$,
with $g_\text{DIF}$, $h_k$, and $\fb_k$
mutually independent,
for each $k \in \{1,2,\ldots, K\}$.
The parameters 
$\sigma_{h_k}^2$ and $\sigma_{f_k}^2$
are the large-scale geometric path losses of 
the link seen by the $k$-th user 
with respect to the transmitter and the RIS, respectively.
The triplet $(g_\text{DIF}, h_k, \fb_k)$ is independent of the noise vectors $\vb_k^{(m)}$, 
$\forall k \in \{1,2,\ldots, K\}$ and $\forall m \in \{0,1,\ldots, M-1\}$. 

\subsection{High-level overview of the CSI acquisition process}
\label{sec:channels}

In the sequel, we assume that, for each $k \in \{1,2,\ldots, K\}$,  
the $k$-th UE has perfect knowledge of its channel gain $c_k^{(m)}$,
$\forall m \in \{0,1,\ldots, M-1\}$, referred to as {\em CSI at the receiver (CSIR)}, 
which can be estimated through 
standard training algorithms, by using 
$N_\text{down-full}$
downlink pilot symbols 
repeated at the beginning of each slot.
In this case, if the power spent by the transmitter
for pilot symbol transmission is significantly greater than the noise variance, 
a single downlink pilot per slot (i.e., $N_\text{down-full}=1$) is sufficient to
ensure that the channel estimation error is quite small \cite{Gursoy.2009}.
On the other hand, we consider two different amount of CSIT:

\begin{enumerate}

\itemsep=0mm

\item
In the case of {\em full} CSIT, the transmitter has 
perfect knowledge of all the relevant channel parameters 
$h_k$ and $\fb_k$, $\forall k \in \{1,2,\ldots, K\}$, and $\gb$. 

\item
In the case of {\em partial} CSIT, 
$\forall m \in \{0, 1, \ldots, M-1\}$, 
the transmitter knows only 
the user index for which 
$|c_k^{(m)}|$ is maximized.
\end{enumerate}

\begin{figure*}[t]
\centering
\includegraphics[width=\linewidth]{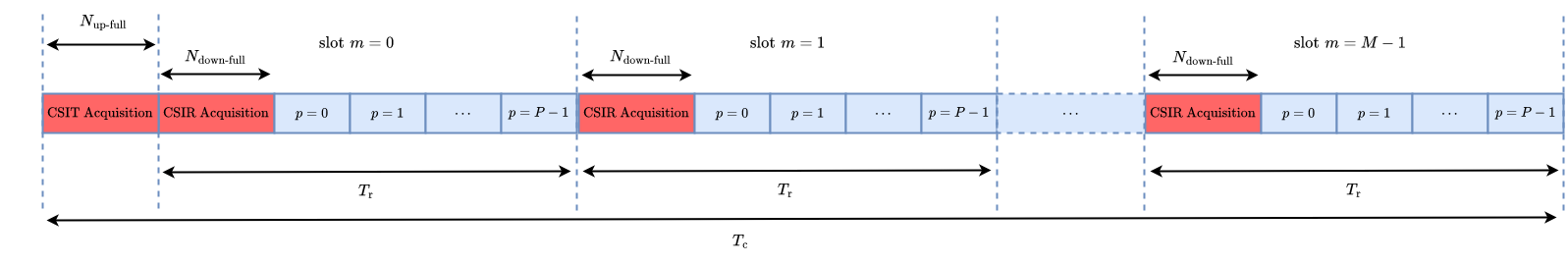}
\caption{According to a TDD protocol, acquisition of full CSIT requires an uplink training phase 
at the beginning of each channel coherence time interval, in addition to 
downlink training phases for CSIR acquisition.
}
\label{fig:fig_2}
\end{figure*}

For acquisition of full CSIT, a viable strategy consists of resorting to 
a {\em time-division duplexing (TDD)} protocol. In this case, 
where channel reciprocity can be exploited, 
acquisition of full CSIT additionally requires a non-trivial 
uplink training session of $N_\text{up-full}$ pilot symbols 
at the beginning of each channel
coherence time  (see Fig.~\ref{fig:fig_2}), 
which might be unrealistic in many
scenarios of practical interest, especially for large values of $K$ and $Q$.
Indeed, the CSIT estimation process requires  
$N_\text{up-full} \ge K (Q+1)$ pilot symbols \cite{Swin.2022}.
\begin{figure*}[t]
\centering
\includegraphics[width=\linewidth]{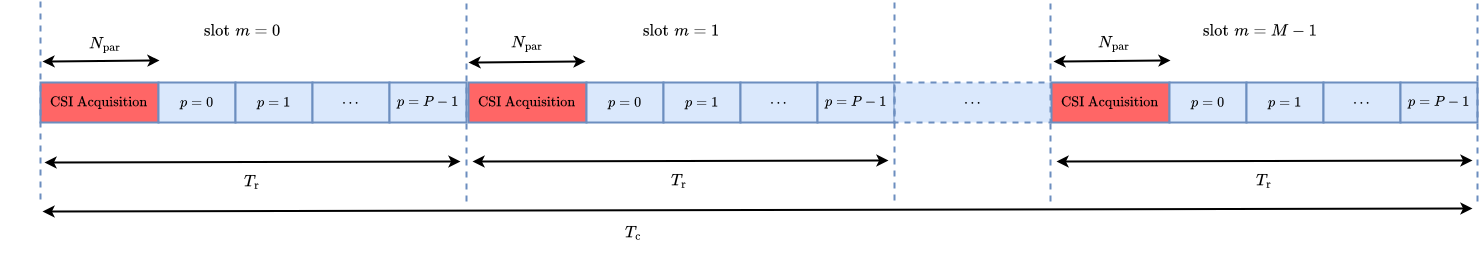}
\caption{Acquisition of partial CSIT can be performed in a centralized or decentralized manner 
by providing at the beginning of each slot 
an initial phase devoted to both CSIR acquisition and discovery of the 
UE with the strongest path for downlink information transmission.
}
\label{fig:fig_3}
\end{figure*}
On the other hand, partial CSIT involves only the discovery of the UE
having the ``best" end-to-end path, which can be identified in 
a centralized or decentralized fashion. 
In the centralized case,  each UE estimates its own 
channel coefficient $c_k^{(m)}$, through, say, a 
common downlink pilot and, then, feeds back the (quantized version of) 
instantaneous channel quality $|c_k^{(m)}|$ to the 
transmitter.\footnote{As a fundamental building component to 
enable the 5G NR system, a physical uplink control channel (PUCCH) is
mainly utilized to convey uplink control information, 
including CSI report for link adaptation and downlink data
scheduling \cite{3GPP-PUCCH}.} 
In such a centralized scenario, 
the duration $N_\text{par}$ (in symbol intervals)  
of the CSI acquisition phase in Fig.~\ref{fig:fig_3}
also depends on the delay in feeding $|c_k^{(m)}|$ back.
If feedback control channels from each UE to the transmitter
are not available, a decentralized CSI acquisition
scheme can be implemented \cite{Bletsas.2006},
according to which, after receiving 
a downlink training packet, each UE will start
its own timer with an initial value that is inversely proportional to
a local estimate of channel quality $|c_k^{(m)}|$. 
The timer of the UE with the strongest end-to-end channel will expire first
and such an UE will transmit an uplink flag packet to the transmitter, 
communicating  its identification number. 
In this decentralized scenario, the amount of 
overhead $N_\text{par}$ involved in selecting the best
UE also depends on the statistics of the wireless channel 
(see \cite{Bletsas.2006} for a probabilistic analysis).

Since the capacity results reported in this paper are valid only when 
the required CSI is perfectly known at the transmitter and at the UEs,
we do not account for channel estimation imperfections hereinafter, by
deferring CSI estimation issues to a future work.  

\section{Sum-rate time-averaged capacity: full CSIT}
\label{sec:sum-rate}

In this section, we assume {\em full} CSIT (see Subsection~\ref{sec:channels}).
The achievable data rate transmitted to the 
$k$-th UE is affected by 
the interference from the signals 
intended for the other users $u \neq k \in \{1,2,\ldots, K\}$
on the same slot.
Assuming that the transmitter encodes the information for 
each user using an i.i.d. Gaussian code,  the {\em time-averaged} rate of user $k$ can be written as 
\be
\overline{\capa}_k = \frac{1}{M} \sum_{m=0}^{M-1} 
\log_2 \left[ 1+ \text{SINR}_k^{(m)} \right]
\quad \text{(in bits/s/Hz)}
\label{eq:Rk}
\ee
for $k \in \{1,2,\ldots, K\}$, with the
signal-to-interference-plus-noise ratio (SINR) at UE $k$ given by
\be
\text{SINR}_k^{(m)} \eqdef \frac{\pot_k^{(m)} \, |c_k^{(m)}|^2}{\displaystyle 
|c_k^{(m)}|^2 \sum_{ \shortstack{\footnotesize $u=1$ \\ \footnotesize $u \neq k$}}^K 
\pot_u^{(m)} +1} \:.
\label{eq:SINR}
\ee
It should be observed that $\capa_k^{(m)} \eqdef \log_2 \left[ 1+ \text{SINR}_k^{(m)} \right]$
in  \eqref{eq:Rk} represents the achievable rate of user $k$ over 
slot $m$ (depending on the current channel conditions), 
which is referred to as the {\em instantaneous per-user rate}.
Accordingly, $\overline{\capa}_k$ is the {\em time-averaged per-user rate} over 
a given channel coherence time. 

The  {\em (net) sum-rate time-averaged capacity} can be expressed as follows
\be
\capaa_\text{sum} =  \max_{
\shortstack{\footnotesize $\pb^{(0)}, \pb^{(1)}, \ldots, \pb^{(M-1)}$ 
\\ \footnotesize $\gammab^{(0)}, \gammab^{(1)}, \ldots,\gammab^{(M-1)}$}}  \,
\xi_\text{full} \sum_{k=1}^K \overline{\capa}_k
\label{eq:max-prob}
\ee
subject to (s.t.)
\barr
& \text{C}_1: 
\frac{1}{M} \sum_{k=1}^K \sum_{m=0}^{M-1}
\pot_k^{(m)} \le \pot_\text{TX}
\nonumber \\ 
& \text{C}_2:
\gamma_q^{(m)} \in \mathcal{R}, \, \text{for $m \in \{0,1,\ldots, M-1\}$} 
\nonumber \\ &
\qquad \qquad \qquad \qquad \text{and  $q \in \{1,2,\ldots, Q\}$} 
\nonumber 
\earr
where $\pb^{(m)} \eqdef [\pot_1^{(m)},\pot_2^{(m)}, \ldots, \pot_K^{(m)}]^\trasp \in \Rset^K$,
for $m \in \{0,1,\ldots,M-1\}$, 
and we recall  again that the channel gain
$c_k^{(m)}$ explicitly depends on the RIS reflection vectors
$\gammab^{(m)} \eqdef [\gamma_1^{(m)}, \gamma_2^{(m)}, \ldots, \gamma_Q^{(m)}]^\herm \in \Cset^{Q}$
of the rapidly time-varying RIS through \eqref{eq:vet-ck}, 
whereas 
\be
\xi_\text{full} \eqdef 
1-\frac{N_\text{up-full}+M \, N_\text{down-full}}{L_\text{c}}
\label{eq:over-full}
\ee
accounts for the pilot overhead required to acquire both CSIT and CSIR  (see Fig.~\ref{fig:fig_2}).
We perform the optimization \eqref{eq:max-prob} in the following two steps, by resorting to
an approach that is called {\em concentration} in the literature 
of estimation theory \cite{Kay.1998}. 

\subsection{First step: optimization with respect to transmit powers}
\label{sec:step-1}

In this subsection, we first find the solution of 
\eqref{eq:max-prob} with respect to 
$\pb^{(0)}, \pb^{(1)}, \ldots, \pb^{(M-1)}$, 
for fixed reflection coefficients
$\Gammab \eqdef [\gammab^{(0)}, \gammab^{(1)}, \ldots,\gammab^{(M-1)}] \in \Cset^{Q \times M}$.

\begin{theorem}
\label{thm:1}
Given the reflection response $\Gammab$ of the rapidly time-varying RIS, 
the slot assignment policy maximizing the sum rate $\xi_\text{full} \sum_{k=1}^K \overline{\capa}_k$ under constraint 
$\text{C}_1$
consists 
of assigning each slot to only one UE having the best channel gain 
for that slot, i.e., for $m \in \{0,1,\ldots, M-1\}$, 
\be
\pot_k^{(m)} =
\begin{cases}
\pot^{(m)} \:, & \text{if $k=k_\text{max}^{(m)}$}
\\
0 \:, & \text{otherwise}
\end{cases}
\label{eq:power-allo}
\ee
s.t. constraint 
\be
\frac{1}{M} \sum_{m=0}^{M-1} \pot^{(m)} \le \pot_\text{TX}
\label{eq:constr-pot-2}
\ee
with  
\be
k_\text{max}^{(m)} \eqdef \arg \max_{k \in \{1,2,\ldots, K\}} |c_k^{(m)}|^2 \: .
\ee
\end{theorem}

\begin{proof}
This result can be proven by using information theoretic arguments \cite{Tse.1997}
or, equivalently, by resorting the principle of mathematical
induction \cite{Jang.2003}.
\end{proof}

Theorem~\ref{thm:1} states that, 
for fixed $\Gammab$, 
the sum-rate time-averaged capacity is
equal to the largest single-user capacity in the system.
The slot assignment policy is a simple 
{\em opportunistic time-sharing} strategy: at each slot,
transmit to the UE with the strongest channel.
This result comes from the fact that, in information theory jargon, 
the model \eqref{eq:rkt} falls into the class of {\em degraded broadcast channels} \cite{Tse-book}, for
which users can be ordered in terms of the quality of signal received in each
time slot $m$, i.e.,  on the basis of $|c_1^{(m)}|^2, |c_2^{(m)}|^2, \ldots, |c_K^{(m)}|^2$.
In this case, there is always one UE experiencing  a stronger channel than
any other UE, which  is able to decode codewords intended for
other UEs, and thus the long-term average
sum-rate is maximized by only transmitting to
such a user.

To derive the sum-rate time-averaged capacity, 
given the reflection response $\Gammab$ of the RIS,  
we have to determine the amount of transmit
power to be allocated to the slots in order to maximize
the sum rate, s.t. constraint \eqref{eq:constr-pot-2}.
When the slot assignment for UEs is
carried out according to 
Theorem~\ref{thm:1}, the multiuser downlink can be viewed
as an opportunistic time-sharing system with dynamic slot allocation, where
UEs receive data through a number of slots assigned
to their own independently. Therefore, given $\Gammab$, the sum-rate time-averaged capacity 
assumes the expression reported in \eqref{eq:max-prob-par}, 
\begin{figure*}[!t]
\normalsize
\barr
\capaa_\text{sum}(\Gammab) = & \max_{\pot^{(0)}, \pot^{(1)}, \ldots, \pot^{(M-1)}}  \,
\frac{\xi_\text{full}}{M} \sum_{m=0}^{M-1}
\log_2\left[1+ \pot^{(m)} \, \alpha_\text{max}(\gammab^{(m)}) \right]\:,
\quad \text{s.t. constraint \eqref{eq:constr-pot-2}}
\label{eq:max-prob-par}
\earr
\hrulefill
\end{figure*}
with  
\be
\alpha_\text{max}(\gammab^{(m)})\eqdef \max_{k \in \{1,2,\ldots, K\}} |c_k^{(m)}|^2 \: .
\label{eq:alphamax}
\ee
The solution of problem \eqref{eq:max-prob-par} is 
{\em water-filling} \cite{Tse-book}
over the slots with the best channel gains among multiple users, i.e, 
for $m \in \{0,1,\ldots, M-1\}$, 
\be
\pot^{(m)} \equiv \pot(\gammab^{(m)})  \eqdef \left[\frac{1}{\lambda}-\frac{1}{\alpha_\text{max}(\gammab^{(m)})}\right]^+
\label{p-wf}
\ee
where $\lambda >0$ is a threshold to be determined
from the total transmit power constraint by solving 
\be
\frac{1}{M} \sum_{m=0}^{M-1} \left[\frac{1}{\lambda}-\frac{1}{\alpha_\text{max}(\gammab^{(m)})}\right]^+ =\pot_\text{TX} \: .
\label{eq:lambda}
\ee
In conclusion, 
the sum-rate time-averaged capacity with fixed $\Gammab$ is given by 
\be
\capaa_\text{sum}(\Gammab) =  
\frac{\xi_\text{full}}{M} \sum_{m=0}^{M-1}
\log_2\left[1+ \pot(\gammab^{(m)}) \, \alpha_\text{max}(\gammab^{(m)}) \right] \:.
\label{eq:C-Gamma}
\ee

\subsection{Second step: optimization with respect to reflection coefficients 
of the rapidly time-varying RIS}
\label{sec_step-2}

At this point, we determine the values of the reflection coefficients
of the rapidly time-varying RIS in order to maximize \eqref{eq:C-Gamma}, s.t. 
constraint $\text{C}_2$.
If channels are constant within a coherence block and their realizations are perfectly known
at the transmitter, it is expected that one can just optimize the reflection
response of the RIS given the perfectly known channel realizations and keep the optimal 
reflection coefficients  constant for the whole coherence block. 
Lemma~\ref{lemma:1} provides a simple proof of this intuition. 

In light of Theorem~\ref{thm:1} and \eqref{eq:C-Gamma}, the optimization problem 
\eqref{eq:max-prob} amounts to solving 
\eqref{eq:max-prob-fin},
\begin{figure*}[!t]
\normalsize
\barr
\capaa_\text{sum} & =   
\max_{\gammab^{(0)}, \gammab^{(1)}, \ldots,\gammab^{(M-1)}}  \,
\frac{\xi_\text{full}}{M} \sum_{m=0}^{M-1}
\log_2\left[1+ \pot(\gammab^{(m)}) \, \alpha_\text{max}(\gammab^{(m)}) \right]
\nonumber \\ & =
\max_{\gammab^{(0)}, \gammab^{(1)}, \ldots,\gammab^{(M-1)}}  \,
\xi_\text{full} \log_2 \left(\left \{ \prod_{m=0}^{M-1} \left[1+ \pot(\gammab^{(m)}) \, \alpha_\text{max}(\gammab^{(m)}) \right]
\right\}^{1/M}\right) \:,
\quad \text{s.t. constraint $\text{C}_2$} 
\label{eq:max-prob-fin}
\earr
\hrulefill
\end{figure*}
where we have also used well-known properties of the logarithm function.

\begin{lemma}
\label{lemma:1}

The sum-rate time-averaged capacity can be achieved by employing a 
slowly time-varying RIS, i.e, 
$M=1$ and $T_\text{r}=T_\text{c}$, 
and, hence, problem \eqref{eq:max-prob-fin} simplifies as follows
\barr
\capaa_\text{sum} = & 
\max_{\gammab \eqdef [e^{j\, \theta_1}, e^{j\, \theta_2}, \ldots, e^{j\, \theta_Q}]^\herm}  \,
\xi_\text{full} \, \log_2 \left[1+\pot(\gammab) \, \alpha_\text{max}(\gammab)\right] 
\nonumber \\ 
& \qquad \text{s.t. $\theta_q \in \Theta$, for $q \in \{1,2,\ldots, Q\}$} 
\label{eq:max-prob-fin-TI}
\earr
with $\Theta \eqdef \left\{\frac{2 \pi}{2^b} \ell, \text{for $\ell \in \{0,1, \ldots, 2^b-1\}$}\right\}$.

\end{lemma}

\begin{proof}
By virtue of the inequality of arithmetic and geometric means \cite{Horn}, it holds that
\begin{multline}
\left \{ \prod_{m=0}^{M-1} \left[1+ \pot(\gammab^{(m)}) \, \alpha_\text{max}(\gammab^{(m)}) \right]
\right\}^{1/M} 
\\
\le \frac{1}{M}
\sum_{m=0}^{M-1} \left[1+ \pot(\gammab^{(m)}) \, \alpha_\text{max}(\gammab^{(m)}) \right]
\end{multline}
with equality if and only if $\gammab \eqdef \gammab^{(0)}=\gammab^{(1)}=\cdots = \gammab^{(M-1)}$.

\end{proof}

Lemma~\ref{lemma:1} infers that, if there is
no constraint on the time-averaged rate of each user, exploitation of the temporal dimension of
the rapidly time-varying RIS is unnecessary, in the sense that time-domain variations of the reflection 
response do {\em not} improve the sum-rate time-averaged capacity, compared to a
slowly time-varying RIS. 
However, as we will show later on,
Lemma~\ref{lemma:1} does not hold if other network requirements 
are taken into account, e.g., optimization without full CSIT
or fairness for the efficient utilization of
scarce radio resources.

The mathematical programming \eqref{eq:max-prob-fin-TI} can be decomposed into the following 
$K$ simpler problems:
\begin{multline}
(\overline{\theta}_{k,1}, \overline{\theta}_{k,2}, \ldots, \overline{\theta}_{k,Q}) 
\\ =
\arg \max_{
\shortstack{\footnotesize $\theta_q \in \Theta$ 
\\ \footnotesize $q \in \{1,2,\ldots,Q\}$}
}  \, \mathcal{F}_k(e^{j \, \theta_1}, e^{j \, \theta_2}, \ldots, e^{j \, \theta_Q})
\label{eq:max-prob-ckm}
\end{multline}
for $k \in \{1,2,\ldots, K\}$,
where, according to \eqref{eq:vet-ck}, we can write
\begin{multline}
\mathcal{F}_k(e^{j \, \theta_1}, e^{j \, \theta_2}, \ldots, e^{j \, \theta_Q})  = 
\left|h_k+ \bg^\herm \, \diag(\fb_k) \, \gammab\right|^2 
\\
= |h_{k}|^2 +  2 \, \Re\left\{
\betab_k^\herm \, \gammab\right\} + \gammab^\herm \, \Bb_k \, \gammab 
\\  = |h_{k}|^2 + 2 \, \Re\left\{\sum_{q=1}^Q \beta_{k,q} \, e^{j \, \theta_q}\right\}
\\ + \sum_{q_1=1}^Q \sum_{q_2=1}^Q \{\Bb_k\}_{q_1,q_2} \, 
e^{j \, \theta_{q_1}} \, e^{-j \, \theta_{q_2}}
\label{eq:ckmsquare}
\end{multline}
where 
$\betab_k \eqdef [\beta_{k,1}, \beta_{k,2}, \ldots, \beta_{k,Q}]^\trasp =h_k \, \diag(\fb_k^*) \, \bg$
and  $\Bb_k \eqdef \diag(\fb_k^*) \, \bg \, \bg^\herm \, \diag(\fb_k) \in \Cset^{Q \times Q}$
is a  Hermitian matrix.
Accordingly, the
sum-rate time-averaged capacity reads as
\be
\capaa_\text{sum}=  
\xi_\text{full} \, \log_2\left(1+ \pot_\text{TX} \, \alpha_\text{opt} \right) 
\label{eq:C-sum}
\ee
where
$\alpha_\text{opt} \eqdef \max_{k \in \{1,2,\ldots,K\}} \,
\mathcal{F}_k(e^{j \, \overline{\theta}_{k,1}}, e^{j \, \overline{\theta}_{k,2}}, \ldots, e^{j \, \overline{\theta}_{k,Q}})$.
Optimization \eqref{eq:max-prob-ckm} is known to be a hard non-convex
NP-hard problem. By exploiting its structure,
problem \eqref{eq:max-prob-ckm} can be transformed into an integer linear
program \cite{Wu.2020}, for which the globally optimal solution can be
obtained by applying the branch-and-bound method. 
Some methods have been developed in \cite{Wu.2020, Zhao.2021,Liang.2022} 
to reduce the computational complexity for the optimal solution, which are based on  
the  block-coordinate descent (or Gauss-Seidel) method \cite[Sec.~2.7]{Bert.1999}. 
Herein, we leverage the successive refinement algorithm proposed 
in \cite[III-B]{Wu.2020}, which is shown to achieve close-to-optimal performance
(see, e.g., the  numerical results reported in \cite{Wu.2020}).

For optimization purposes, it is convenient to extrapolate from $\mathcal{F}_k(e^{j \, \theta_1}, 
e^{j \, \theta_2}, \ldots, e^{j \, \theta_Q})$
the contribution of the single reflection phase $\theta_\nu$, for 
$\nu \in \{1,2, \ldots, Q\}$.  To this aim, after some algebraic rearrangements, it follows
from \eqref{eq:ckmsquare} that
\begin{multline}
\mathcal{F}_k(e^{j \, \theta_1}, \ldots, e^{j \, \theta_{\nu-1}}, 
e^{j \, \theta_\nu}, e^{j \, \theta_{\nu+1}}  
  \ldots, e^{j \, \theta_Q}) 
  \\ = \lambda_{k,\nu} + 2\, |\chi_{k,\nu}| \, 
  \cos(\theta_\nu+\measuredangle \chi_{k,\nu})
\label{eq:cos}
\end{multline}
where both complex quantities
\barr
\lambda_{k,\nu} & =
|h_{k}|^2 + 2 \, \Re\left\{\sum_{
\shortstack{\footnotesize $q=1$ 
\\ \footnotesize $q \neq \nu$}}^Q \beta_{k,q} \, e^{j \, \theta_q}\right\}
\nonumber \\ & \qquad 
+ \sum_{\shortstack{\footnotesize $q_1=1$ 
\\ \footnotesize $q_1 \neq \nu$}}^Q \sum_{\shortstack{\footnotesize $q_2=1$ 
\\ \footnotesize $q_2 \neq \nu$}}^Q \{\Bb_k\}_{q_1,q_2} \, 
e^{j \, \theta_{q_1}} \, e^{-j \, \theta_{q_2}}
\\
\chi_{k,\nu} & = \beta_{k,\nu} +
\sum_{\shortstack{\footnotesize $q=1$ 
\\ \footnotesize $q \neq \nu$}}^Q  \{\Bb_k\}_{\nu,q} \, 
e^{-j \, \theta_{q}}
\earr
do {\em not} depend on the reflection phase $\theta_\nu$.
Preliminarily, let us consider the problem of optimizing  
with respect to the variable $\theta_\nu \in \Theta$ while
holding all other variables constant, i.e., 
\be
\overline{\theta}_{k,\nu}=
\arg \max_{\theta_\nu \in \Theta}
  \, \mathcal{F}_k(e^{j \, \theta_1}, \ldots, e^{j \, \theta_{\nu-1}}, 
e^{j \, \theta_\nu}, e^{j \, \theta_{\nu+1}}  
  \ldots, e^{j \, \theta_Q})
\label{eq:max-prob-ckm-one}
\ee
for $k \in \{1,2,\ldots, K\}$ and 
$\nu \in \{1,2, \ldots, Q\}$. Accounting for \eqref{eq:cos},
problem \eqref{eq:max-prob-ckm-one} boils down to
\be
\overline{\theta}_{k,\nu}=
\arg \max_{\theta_\nu \in \Theta} \, 
\cos(\theta_\nu+\measuredangle \chi_{k,\nu}) \:.
\label{eq:max-prob-ckm-one-search}
\ee
Since \eqref{eq:max-prob-ckm-one-search} has at least one solution, 
problem \eqref{eq:max-prob-ckm} can be solved 
by optimizing with respect to the first variable $\theta_1$
while holding all
other variables constant, then optimizing with respect to the second variable $\theta_2$,
while holding all
other variables constant, and so on. This is referred
to as the block-coordinate ascent algorithm, whose convergence can be
shown under relatively general conditions \cite{Wu.2020}.
Specifically, given the current iterate
$\overline{\thetab}_k(i) \eqdef [\overline{\theta}_{k,1}(i), \overline{\theta}_{k,2}(i), 
\ldots, \overline{\theta}_{k,Q}(i)]^\trasp \in \Rset^{Q} $, 
the block-coordinate ascent algorithm generates the next iterate
$\overline{\thetab}_k(i+1)$ according to the iteration
shown in \eqref{eq:max-prob-ckm-one-iter},
\begin{figure*}[!t]
\normalsize
\barr
\overline{\theta}_{k,\nu}(i+1)  & =
\arg \max_{\theta_\nu \in \Theta}
  \, \mathcal{F}_k(e^{j \, \overline{\theta}_{k,1}(i+1)}, \ldots, 
  e^{j \, \overline{\theta}_{k,\nu-1}(i+1)}, 
e^{j \, \theta_\nu}, e^{j \, \overline{\theta}_{k,\nu+1}(i)}  
  \ldots, e^{j \, \overline{\theta}_{k,Q}(i)}
  \nonumber \\ & =
\arg \max_{\theta_\nu \in \Theta} \, 
\cos[\theta_\nu+\measuredangle \chi_{k,\nu}(i+1)]
\label{eq:max-prob-ckm-one-iter}
\earr
\hrulefill
\end{figure*}
with 
\begin{align}
\chi_{k,\nu}(i+1)  = \beta_{k,\nu} +
\sum_{q=1}^{\nu-1}  \{\Bb_k\}_{\nu,q} \, 
e^{-j \, \overline{\theta}_{k,q}(i+1)} 
\\ +
\sum_{q=\nu+1}^{Q}  \{\Bb_k\}_{\nu,q} \, 
e^{-j \, \overline{\theta}_{k,q}(i)} \: .
\label{eq:chiknu}
\end{align}

The optimization procedure of the reflection
response of the slowly time-varying RIS 
is summarized as Algorithm~\ref{table:tab_1}.
It is noteworthy that such an optimization algorithm 
explicitly accounts for the discrete nature of the reflection
coefficients of the RIS. Moreover,  
it involves only simple operations and does
not require any step size selection. 

\begin{algorithm}[t]
\caption{The proposed optimization of the reflection response of the RIS.}
\label{table:tab_1}
\begin{algorithmic}
\STATE  {\bf Input quantities}: $\betab_k $ and $\Bb_k$, $\forall k \in \{1,2,\ldots, K\}$. 

\STATE
{\bf Output quantities}: estimates of the solutions $\overline{\theta}_{k,1}, \overline{\theta}_{k,2}, \ldots, \overline{\theta}_{k,Q}$ of \eqref{eq:max-prob-ckm},  $\forall k \in \{1,2,\ldots, K\}$.

\STATE
{\bf Initialization}: 
$\overline{\theta}_{k,2}(0)=\overline{\theta}_{k,3}(0)=\cdots=
\overline{\theta}_{k,Q}(0)=0$, $\forall k \in \{1,2,\ldots, K\}$.

\STATE
{\bf Stopping criterion}:
identify a suitable maximum number of iteration $I_\text{max}$.
 
\begin{enumerate}[1.]

\itemsep=1mm

\item
Set $k=1$.

\item
Set $i=0$.

\item
Set $\nu=1$.

\item
Evaluate $\chi_{k,\nu}(i+1)$ using \eqref{eq:chiknu}.

\item
Compute $\overline{\theta}_{k,\nu}(i+1)$ using 
\eqref{eq:max-prob-ckm-one-iter} by an exhaustive searching algorithm.

\item
Set $\nu=\nu+1$: if $\nu < Q$, then go back to Step 4.

\item
If  $i < I_\text{max}$, then set $i=i+1$ and
go back to Step 3.

\item
Set $k=k+1$: if $k < K$, then go back to Step 2, else end the procedure.

\end{enumerate}
\end{algorithmic}
\end{algorithm}

As a system performance metric, by applying  channel coding across
channel coherence  intervals (i.e., over an ``ergodic" 
interval of channel variation with time), the ensemble average of the 
sum-rate time-averaged capacity \eqref{eq:C-sum} is given by
\be
\Es[\capaa_\text{sum}] = \xi_\text{full} 
\int_{0}^{+\infty} 
\log_2\left(1+ \pot_\text{TX} \, \alpha \right) \, f_{\alpha_\text{opt}}(\alpha) 
\, {\rm d} \alpha 
\label{eq:Ave-C-sum}
\ee
where $f_{\alpha_\text{opt}}(\alpha)$ is 
the probability density function (pdf) of the non-negative random variable
$\alpha_\text{opt}$.
Analytical evaluation of \eqref{eq:Ave-C-sum} is complicated 
by the fact that there is no closed-form solution of the optimization
problem \eqref{eq:max-prob-ckm}. Therefore, the dependence of 
$\Es[\capaa_\text{sum}]$ on the main system parameters (e.g., $K$ and $Q$)
will be numerically studied in Section~\ref{sec:simulations}.

\section{Rapidly time-varying RIS: Partial CSIT}
\label{sec:partial-CSI}

Achievement of the sum-rate time-averaged capacity  \eqref{eq:C-sum} 
requires optimization of the slowly time-varying RIS reflection response
at a rate $1/T_\text{c}$ approximately equal to the Doppler spread of the
underlying channels, which in turn 
involves acquisition of full CSIT.
In this section, 
we propose a scheme that constructs 
random reflections at the RIS in both space and 
time domains in order to 
take advantage of the presence
of the rapidly time-varying RIS in the system and to perform multiuser 
diversity scheduling by assuming partial CSIT (see Subsection~\ref{sec:channels}).
It should be observed that 
the performance of randomized slowly time-varying RIS
has been evaluated in \cite{Tegos.2022}. However, since  
the reflection  coefficients of the RIS are randomly varied at a rate 
$1/T_\text{c}$,  the beneficial effects of randomized rapidly time-varying reflections 
in terms of multiuser diversity and network fairness have not been 
studied.

Herein, we choose the reflection 
coefficients $\gamma_q^{(m)}=e^{j \, \Theta_q^{(m)}}$ in \eqref{eq:gamma},
for $q \in \{1,2, \ldots, Q\}$ and $m \in \{0,1,\ldots, M-1\}$, 
as a sequence of i.i.d. random variables with respect to 
both the space index $q$ and the time index $m$, where
each random variable $\Theta_q^{(m)}$ assumes 
equiprobable values in $\Theta$ (defined in 
Lemma~\ref{lemma:1}).
In this case, besides constructing the
spatially varying phase profile, 
the phase responses of the meta-atoms are randomly varied 
in time from one slot to another. 
After each coherence interval, we independently
choose another sequence of reflection 
coefficients, and so forth. 
In this scenario, Theorem~\ref{thm:1} still holds 
and the optimal scheduling strategy consists of 
transmitting in each slot to the user with 
the strongest channel. 
Since only partial CSIT is assumed, 
we allocate equal powers at all
slots, i.e., we set $\pot^{(m)}=\pot_\text{TX}$ in \eqref{eq:constr-pot-2},
$\forall m \in \{0,1,\ldots, M-1\}$.  
This is particularly relevant for the downlink, where the
base station can operate always at its peak total power.
Such a strategy 
is almost optimal  at high SNR \cite{Tse-book}, i.e., 
for sufficiently large values of $\pot_\text{TX}$.
Therefore, the corresponding (net) sum-rate time-averaged capacity is 
\be
\capaa_\text{rand} =
\frac{\xi_\text{par}}{M} \sum_{m=0}^{M-1}
\log_2\left[1+ \pot_\text{TX} \, \alpha_\text{max}(\gammab^{(m)}) \right] 
\label{eq:Ave-C-sum-rand}
\ee
with $\alpha_\text{max}(\gammab^{(m)})$ defined in 
\eqref{eq:alphamax}, where the reflection vector 
$\gammab^{(m)} = [e^{j \, \Theta_1^{(m)}},
e^{j \, \Theta_2^{(m)}}, \ldots, e^{j \, \Theta_Q^{(m)}}]^\herm$
is chosen according to the 
proposed rapidly time-varying randomized scheme,
whereas 
\be
\xi_\text{par} \eqdef 
1-\frac{M \, N_\text{par}}{L_\text{c}}=1-\frac{ N_\text{par}}{P}
\label{eq:over-par}
\ee
accounts for the overhead required to acquire CSI  (see Fig.~\ref{fig:fig_3}).
As expected, it results that  $\capaa_\text{rand} \le \capaa_\text{sum}$.
It is noteworthy that, by varying the phases of the meta-atoms in both 
space and time from one slot to another, 
the RIS randomly scatters multiple reflected beams, 
whose number, amplitudes, and directions change
on a slot-by-slot basis.
Therefore, with partial CSI only, 
the transmitter can schedule to the user currently 
closest to one of such beams that ensures maximum signal power.  
With many users, there
is likely to be a user very close to a ``strong" reflected beam at any slot.
Such an intuition is formally justified in the next subsection.

\subsection{Asymptotic performance analysis}
\label{sec:dis-rand}

Herein, we evaluate the ensemble average  
$\Es[\capaa_\text{rand}]$, which is defined as
the expected value of \eqref{eq:Ave-C-sum-rand}, over the
sample space of the $4$-tuple $\{\bh, \gb, \Fb, \Gammab\}$,
with $\bh \eqdef [h_1, h_2, \ldots, h_K]^\trasp \in \Cset^K$, 
$\Fb \eqdef [\fb_1,  \fb_2, \ldots,  \fb_K] \in \Cset^{Q \times K}$,
and $\Gammab$ has been defined at the beginning of 
Subsection~\ref{sec:step-1}. 
By virtue of the conditional expectation rule \cite{Proakis} and according to the 
proposed random reflection scheme, one has
\eqref{eq:Capa-ave-rand}, 
\begin{figure*}[!t]
\normalsize
\barr
\Es[\capaa_\text{rand}] & = \Es_{\bh, \gb, \Fb} \left\{ 
\Es_{\Gammab} \left[ \capaa_\text{rand} \, | \, \bh, \gb, \Fb \right] \right\}
 =
\frac{\xi_\text{par}}{M} \sum_{m=0}^{M-1}
\Es_{\bh, \gb, \Fb} \left\{ 
\Es_{\Gammab} \left[ 
\log_2\left(1+ \pot_\text{TX} \, \alpha_\text{max}(\gammab^{(m)}) \right) \, | \, \bh, \gb, \Fb 
\right] \right\}  
\nonumber \\ & =
\xi_\text{par} \, \Es_{\bh, \gb, \Fb} \left\{ 
\Es_{\Gammab} \left[ 
\log_2\left(1+ \pot_\text{TX} \, \alpha_\text{max}(\gammab^{(m)}) \right) \, | \, \bh, \gb, \Fb 
\right] \right\}
\nonumber \\ & =
\xi_\text{par} \, \left(\frac{1}{2^b}\right)^Q 
\sum_{\ell_1=0}^{2^b-1} \sum_{\ell_2=0}^{2^b-1} \cdots \sum_{\ell_Q=0}^{2^b-1}
\Es_{\bh, \gb, \Fb} \left\{\log_2\left[1+ \pot_\text{TX} \, \alpha_\text{rand}\right] \right\}
\label{eq:Capa-ave-rand}
\earr
\hrulefill
\end{figure*}
where 
\be
\alpha_\text{rand} \eqdef \max_{k \in \{1,2,\ldots,K\}} \,
\left|h_{k} + \gb^\herm \, \Thetab^* 
\, \fb_{k}\right|^2
\label{eq:alpha_rand_rd-los}
\ee
with
\be
\Thetab \eqdef \diag\left[e^{j \, \frac{2 \pi}{2^b} \ell_1},
e^{j \, \frac{2 \pi}{2^b} \ell_2}, \ldots, e^{j \, \frac{2 \pi}{2^b} \ell_Q}\right] \:.
\label{eq:Theta-TV}
\ee

For the sake of analysis, 
we assume that the channel between the transmitter and the RIS is characterized
by a dominant LoS component, i.e., $\gb =\gb_\text{LOS}$.
The LoS assumption is reasonable if the distance between the BS and the RIS
is small compared to their height above ground, 
without any sort of an obstacle between them.
The case in which the channels between the transmitter and the meta-atoms of the RIS
have a diffuse component will be numerically studied in Section~\ref{sec:simulations}.
For simplicity, we also consider the case 
in which the users approximately experience the same large-scale 
geometric path loss, i.e., the parameters $\sigma_{h_k}^2$ and $\sigma_{f_k}^2$
do not depend on $k$, i.e., $\sigma_{h_k}^2 \equiv \sigma_{h}^2$ and
$\sigma_{f_k}^2 \equiv \sigma_{f}^2$, $\forall k \in \{1,2,\ldots, K\}$,
which will be referred to as the case of {\em homogeneous} users.
From
a physical viewpoint, this happens when the users form a 
{\em cluster}, wherein the distances between the different UEs are negligible
with respect to the distance between the transmitter and the RIS.
A numerical performance analysis when the users  
have {\em heterogeneous} path losses will be provided in Section~\ref{sec:simulations}.
Under the above assumptions, it is readily seen that the random variable 
$\alpha_\text{rand}$ turns out to be the maximum of $K$ i.i.d. exponential 
random variables $Y_k$ with mean $\sigma_h^2 + \sigma_f^2 \, \sigma_g^2 \, Q$
and, thus, 
the pdf of $\alpha_\text{rand}$ is given by 
\begin{multline}
f_{\alpha_\text{rand}}(\alpha)  = 
K \,  f_{Y_k}(\alpha)  
\, \left[F_{Y_k}(\alpha)\right]^{K-1} 
\\ = 
\frac{K}{\sigma_h^2 + \sigma_f^2 \, \sigma_g^2 \, Q} 
\, e^{-\frac{\alpha}{\sigma_h^2 + \sigma_f^2 \, \sigma_g^2 \, Q}} \, 
\\ \cdot  
 \left[1- e^{-\frac{\alpha}{\sigma_h^2 + \sigma_f^2 \, \sigma_g^2 \, Q}}\right]^{K-1} \:,
 \quad \text{for $\alpha \ge 0$} 
\label{eq:frand}
\end{multline}
where $f_{Y_k}(x)$ and 
$F_{Y_k}(x)$ denote the pdf and  
the cumulative distribution function (cdf) of the random variable $Y_k$.
Since the distribution of the random variables $Y_1, Y_2, \ldots, Y_K$ does not 
depend on the indexes $\ell_1, \ell_2, \ldots, \ell_Q$, the average sum-rate capacity
\eqref{eq:Capa-ave-rand} simplifies to
\be
\Es[\capaa_\text{rand}]= \xi_\text{par}
\int_{0}^{+\infty} 
\log_2\left(1+ \pot_\text{TX} \, \alpha \right) \, f_{\alpha_\text{rand}}(\alpha) 
\, {\rm d} \alpha \:.
\label{eq:Ave-C-sum-rand}
\ee 

Trying to work with \eqref{eq:frand} for evaluating \eqref{eq:Ave-C-sum-rand} is difficult even for
small values of $K$.
For such a reason, we apply extreme value theory \cite{Lea.1978}
to calculate the distribution of $\alpha_\text{rand}$
when $K$ is sufficiently large (i.e., in practice greater than $10$).
It can be shown that, as $K \to + \infty$, the random 
variable $\alpha_\text{rand}$ convergences in distribution 
to the Gumbel function \cite{Gumbel.1958}, i.e., 
\be
\lim_{K \to \infty} F_{\alpha_\text{rand}}(\alpha) =
\text{exp}\left[ - e^{- \frac{\alpha-b_K}{a_K}} \right]
\label{eq:limitdis}
\ee
where $F_{\alpha_\text{rand}}(\alpha)$
is the cdf of  $\alpha_\text{rand}$, and the normalizing
constants $b_K$ and $a_K$ are given by 
\barr
b_K & \eqdef F_{Y_k}^{-1}\left(1-\frac{1}{K}\right) =
(\sigma_h^2 + \sigma_f^2 \, \sigma_g^2 \, Q) \, \ln K
\\
a_k & \eqdef  
\frac{1}{K \, f_{Y_k}(b_K)} = \sigma_h^2 + \sigma_f^2 \, \sigma_g^2 \, Q \: .
\earr
The proof of this result relies on the limit laws for maxima \cite{Lea.1978}
and the fact that the cdf of the exponential random variable $Y_k$ is 
a von Mises function \cite{Kaly.2012}, i.e., 
\be 
\lim_{\alpha \to + \infty} 
\left[ \frac{1-F_{Y_k}(\alpha)}{f_{Y_k}^2(\alpha)}\right]
\, \frac{{\rm d}}{{\rm d}\alpha} f_{Y_k}(\alpha)=-1 \: .
\ee
On the basis of \eqref{eq:limitdis}, the average sum-rate capacity \eqref{eq:Ave-C-sum-rand}
can be explicated as
\begin{multline}
\Es[\capaa_\text{rand}] \asymp \frac{\xi_\text{par}}{a_K}
\int_{0}^{+\infty} 
\log_2\left(1+ \pot_\text{TX} \, \alpha \right) 
\\ \cdot 
e^{- \frac{\alpha-b_K}{a_K}} \, \text{exp}\left[ - e^{- \frac{\alpha-b_K}{a_K}} \right]
\, {\rm d} \alpha 
\label{eq:Ave-C-sum-asympt}
\end{multline}
with $x \asymp y$ indicating that $\lim_{K \to + \infty} x/y =1$.
By using the Maclaurin series of the exponential function,  
the integral \eqref{eq:Ave-C-sum-asympt} can be 
rewritten as an absolutely convergent series \cite[Appendix~C]{Kaly.2012}, which can be
approximately evaluated by using a finite number of terms. 

Let  $\rho_\text{rand} \eqdef \pot_\text{TX} \, \alpha_\text{rand}$ be
the receive SNR, whose
average value can be calculated for large $K$ by using the Gumbel distribution
\eqref{eq:limitdis}, thus obtaining
\barr
\Es[\rho_\text{rand}] & = b_K + C \, a_K
\nonumber \\ & =
\pot_\text{TX} \left[ (\sigma_h^2 + \sigma_f^2 \, \sigma_g^2 \, Q) \, \ln K
+ C \, 
\left(\sigma_h^2 + \sigma_f^2 \, \sigma_g^2 \, Q\right) \right] 
\label{eq:ASNR-rand}
\earr
with $C \approx 0.5772$ being the Euler-Mascheroni constant.
For sufficiently large values of $K$,  
the mean overall channel 
logaritmically scales with the number of user $K$: 
this implies that $\Es[\capaa_\text{rand}] $ increases double logarithmically in $K$, 
and, thus, 
asymptotically, our randomized scheme does {\em not} suffer of 
a loss in multiuser diversity \cite{Tse-book}.
However, there is only a {\em linear} growth of $\Es[\rho_\text{rand}]$ with the
number $Q$ of meta-atoms: this implies that $\Es[\capaa_\text{rand}] $ increases  
logarithmically in $Q$.

\section{Fair Scheduling algorithms}
\label{sec:PFS}

The scheme discussed in Section~\ref{sec:sum-rate} focuses
on maximizing the sum rate time-averaged capacity, i.e., finding 
the solution of the optimization problem \eqref{eq:max-prob}. 
However, fairness among users is an important practical issue
when the users experience different fading effects.
A simple yet informative index to quantify 
the fairness of a scheduling scheme was proposed in \cite{Jain.1984},
which can be expressed as follows
\be
\fair \eqdef \frac{\displaystyle \left(\sum_{k=1}^K \overline{\capa}_k\right)^2}
{\displaystyle K \sum_{k=1}^K \overline{\capa}_k^2} \: .
\label{eq:fair-index}
\ee
If all the users get the same time-averaged rate, i.e., the rates
$\overline{\capa}_k$ are all equal, then the fairness index is $\fair=1$
and the scheduler is said to be $100\%$ fair. On the other hand,
a scheduling scheme that favors only a few selected users
has a fairness index $\fair \to 0$. 

The scheme developed in Section~\ref{sec:sum-rate},  
which exploits full CSIT and employs a slowly time-varying optimized RIS,
exhibits a fairness index $\fair_\text{sum}=1/K$ over the time-scale $T_\text{c}$, 
since  it supports only one user at each channel coherence time.
As expected, in this case, the scheduler becomes 
increasingly unfair as the number $K$ of users increases.
On the other hand, the suboptimal scheme proposed in 
Section~\ref{sec:partial-CSI}, 
which relies on partial CSIT and makes use of 
a randomized rapidly time-varying RIS, might ensure a fairness index
$\fair_\text{rand}$ over the time-scale $T_\text{c}$ 
that is greater than or equal to $1/K$, i.e., 
$\fair_\text{rand} \ge \fair_\text{sum}$. 
This behavior will be numerically shown in Section~\ref{sec:simulations}
and, intuitively, it comes from the fact that more than one
user might be scheduled for  each channel coherence time interval, due to the time variability 
induced by the rapidly time-varying RIS. 
Motivated by such an intuition, our goal is herein to demonstrate that exploitation
of the temporal dimension of the rapidly time-varying RIS is useful for 
developing scheduling strategies that explicitly take into account fairness
over a time horizon equal to the channel coherence time.

\subsection{Proportional fair scheduling with rapidly time-varying  RIS}

We focus on proportional fair scheduling (PFS) originally proposed in the network scheduling context
by \cite{Kelly.1998} as an alternative for a max-min scheduler, since it 
offers an attractive trade-off between the
maximum time-averaged throughput and user fairness.\footnote{The usefulness 
of rapidly time-varying RIS can be more generally shown
by considering other network utility functions that 
achieve a desired
balance between network-wide overall performance and user
fairness \cite{Mo.2000}.} 
Specifically, we start from introducing the 
{\em exponentially weighted moving average} rate  of user $k$, which can be written as 
\be
\overline{\capa}_k^{(m)} = \left(1-\frac{1}{m+1}\right)  \, \overline{\capa}_k^{(m-1)}
+ \frac{1}{m+1}  \, \capa_k^{(m)}
\quad \text{(in bits/s/Hz)}
\label{eq:EWMA}
\ee
for $k \in \{1,2,\ldots, K\}$ and $m \in \{1, 2, \ldots, M-1\}$, whereas 
$\overline{\capa}_k^{(0)}=\capa_k^{(0)}$.
Eq.~\eqref{eq:EWMA} is recursive, since the current average rate $\overline{\capa}_k^{(m)}$ 
at the time slot $m$ is calculated using the average achievable data rate that has been
benefited from user $k$ till the time slot $m$. One might arrange $\overline{\capa}_k^{(m)}$
to show that it is the weighted average of all the preceding $m$ instantaneous per-user rates
$\capa_k^{(m)}, \capa_k^{(m-1)}, \ldots, \capa_k^{(0)}$, where 
the weight $w_k^{(m)}(i)$ of the instantaneous rate $\capa_k^{(m-i)}$ is given by
\be
w_k^{(m)}(i) = \frac{1}{m+1} \, \left(1-\frac{1}{m+1}\right)^i
\ee
for $i \in \{1, 2, \ldots, m\}$, with $m \in \{1, 2, \ldots, M-1\}$. Therefore, 
recursion \eqref{eq:EWMA} assigns lower weights to the older instantaneous rates
that are achievable over the time-scale $T_\text{c}$.
Such a temporal smoothing introduces a memory in the resource allocation procedure
over each channel coherence time. 
To achieve a desired balance between network-wide overall performance and user
fairness, we consider the following {\em per-time-slot} PFS optimization problem: 
\be
\max_{\pb^{(m)}, \gammab^{(m)}}  \,
 \sum_{k=1}^K \log_2 \left[\overline{\capa}_k^{(m)}\right] \:, 
\quad \text{s.t. $\text{C}_1$ and $\text{C}_2$}
\label{eq:max-prob-w}
\ee
for $m \in \{0, 1, \ldots, M-1\}$, 
where the constraints $\text{C}_1$ and $\text{C}_2$ have been 
reported in \eqref{eq:max-prob}. 
It is noteworthy that, since problem  \eqref{eq:max-prob-w} has to 
be solved on a slot-by-slot basis, 
its solution is inherently time varying over the channel coherence time interval, i.e., 
the optimal values of the power vector $\pb^{(m)}$ and 
the reflection vector $\gammab^{(m)}$ might vary from
one time slot to another. Hence, a slowly time-varying RIS (see Subsection~\ref{sec:TI})
could not be useful and exploitation of the temporal dimension of the RIS may be
necessary. 
Under the assumption that $\left | \overline{\capa}_k^{(m)}-\overline{\capa}_k^{(m-1)}\right|$
is sufficiently small, by invoking Taylor's theorem, one can resort to the
linear approximation of $\log_2 \left[\overline{\capa}_k^{(m)}\right]$ 
for $ \overline{\capa}_k^{(m)}$ near the point $ \overline{\capa}_k^{(m-1)}$, 
thus yielding
\barr
\log_2 \left[\overline{\capa}_k^{(m)}\right] & \approx \log_2 \left[\overline{\capa}_k^{(m-1)}\right]
\nonumber \\ & 
+ \frac{1}{\overline{\capa}_k^{(m-1)} \, \ln(2)}
\left[\overline{\capa}_k^{(m)}-\overline{\capa}_k^{(m-1)}\right]
\nonumber \\ & =\log_2 \left[\overline{\capa}_k^{(m-1)}\right] + 
\frac{1}{(m+1) \, \ln(2)} \left[\frac{\capa_k^{(m)}}{\overline{\capa}_k^{(m-1)}}\right]
\label{eq:logapprox}
\earr
where the equality comes from the application of \eqref{eq:EWMA}. Capitalizing on 
\eqref{eq:logapprox}, the optimization problem \eqref{eq:max-prob-w} can be
approximated as follows
\be
\max_{\pb^{(m)}, \gammab^{(m)}}  \,
\sum_{k=1}^K \frac{\capa_k^{(m)}}{\overline{\capa}_k^{(m-1)}} \:, 
\quad \text{s.t. $\text{C}_1$ and $\text{C}_2$}
\label{eq:max-prob-w-approx-1}
\ee
for $m \in \{0, 1, \ldots, M-1\}$. 
Strictly speaking, under scheduling \eqref{eq:max-prob-w-approx-1}, 
users compete in each time slot not directly based on their instantaneous 
per-user rates but based on such rates normalized by their respective average rates.
For a fixed reflection vector $\gammab^{(m)}$, 
following the same information-theoretic arguments invoked 
in Subsection~\ref{sec:step-1} (see, in particular, Theorem~\ref{thm:1}), maximization
of $\sum_{k=1}^K {\capa_k^{(m)}}/{\overline{\capa}_k^{(m-1)}}$
with respect to $\pb^{(m)}$, s.t. constraint $\text{C}_1$, leads to the
opportunistic time-sharing
power allocation scheme \eqref{eq:power-allo}, s.t. \eqref{eq:constr-pot-2}, with 
$k_\text{max}^{(m)}$ replaced by 
\be
k_\text{pfs}^{(m)} \eqdef \arg \max_{k \in \{1,2,\ldots, K\}} 
|c_k^{(m)}|^{\frac{2}{\overline{\capa}_k^{(m-1)}}}
\ee
where we have also used the high-SNR approximation 
\be
\capa_k^{(m)}/\overline{\capa}_k^{(m-1)} \approx 
\log_2 \left[ \left \{ \text{SINR}_k^{(m)} \right \}^{\frac{1}{\overline{\capa}_k^{(m-1)}}} \right]
\ee
with $\text{SINR}_k^{(m)}$ defined in \eqref{eq:SINR}.
The optimal values of $\pot_k^{(m)}$ in \eqref{eq:power-allo}
and $\gammab^{(m)}$, which is embedded in $|c_k^{(m)}|$ [see \eqref{eq:vet-ck}], 
depend on the available CSIT. 

\subsubsection{Full CSIT}
\label{sec:PFS-fullCSIT}
In this case, along the same lines of  Section~\ref{sec:sum-rate}, 
the power allocated at the user $k =k_\text{pfs}^{(m)} $ during the 
$m$-th time slot is given by \eqref{p-wf} with $\alpha_\text{max}(\gammab^{(m)})$
replaced by 
\be
\alpha_\text{pfs}(\gammab^{(m)})\eqdef \max_{k \in \{1,2,\ldots, K\}} 
|c_k^{(m)}|^{\frac{2}{\overline{\capa}_k^{(m-1)}}} \: .
\label{eq:alphamax-pfs}
\ee
In such a case, the values of the reflection coefficients of the rapidly time-varying RIS 
can be obtained by a straightforward modification of Algorithm~\ref{table:tab_1}
(details are omitted). 

\subsubsection{Partial CSIT}
\label{sec:PFS-partCSIT}
When there is knowledge of partial CSIT only, as in 
Section~\ref{sec:partial-CSI}, we allocate power uniformly over the time slots
by setting 
$\pot^{(m)}=\pot_\text{TX}$ in \eqref{eq:constr-pot-2},
$\forall m \in \{0,1,\ldots, M-1\}$.
Moreover, since the reflection coefficients of the rapidly time-varying RIS
cannot be optimized without full CSIT, we choose 
$\gammab^{(m)}$ according to the same 
randomized algorithm proposed in  
Section~\ref{sec:partial-CSI}. For each time slot, the users are
scheduled through \eqref{eq:alphamax-pfs}.

It is noteworthy that, by virtue of the proposed algorithms, a rapidly time-varying RIS allows to 
achieve fairness over the time-scale $T_\text{c}$. 
PFS can be also implemented with a slowly time-varying RIS.
However, in this case, since the reconfigurability rate of the RIS is equal to $1/T_\text{c}$, 
fairness is achieved on a time horizon that is greater than $T_\text{c}$
and, consequently, users experience longer delays 
between successive transmissions.

\begin{figure}[t]
\centering
\includegraphics[width=\columnwidth]{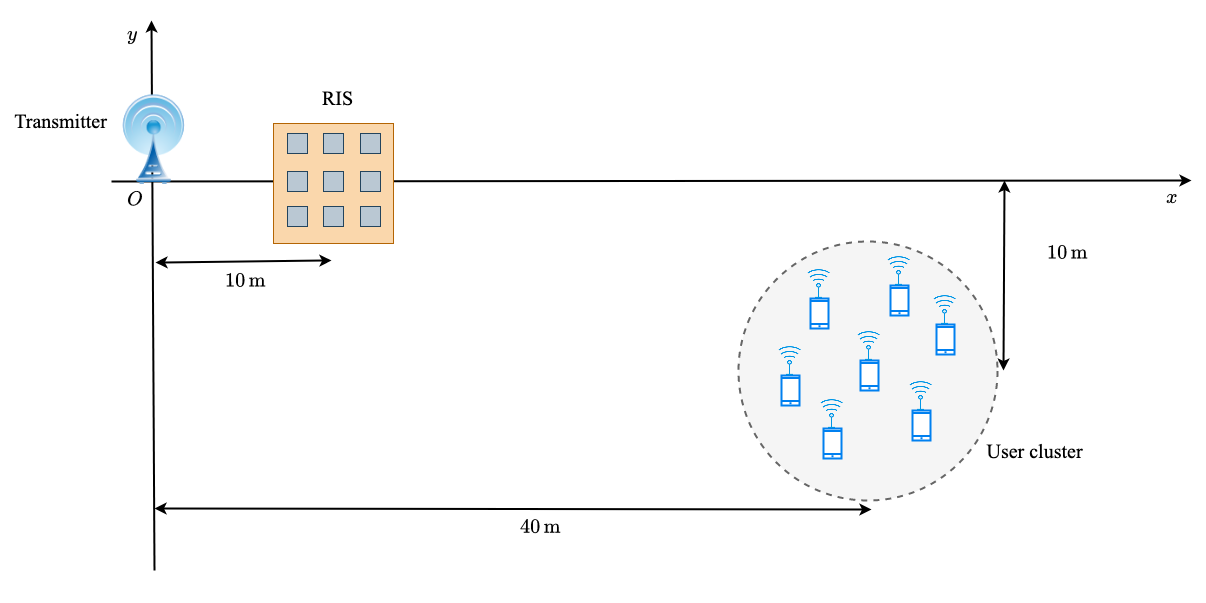}
\caption{Simulation setup of the considered downlink multiuser system.}
\label{fig:fig_simul}
\end{figure}

\section{Numerical results}
\label{sec:simulations}

In this section, with reference to 
the considered RIS-aided multiuser downlink, 
we report Monte Carlo numerical results aimed at 
validating the proposed designs and completing the developed 
performance analysis.  
With reference to Fig.~\ref{fig:fig_simul}, we consider a two-dimensional Cartesian system, wherein the 
transmitter and the center of the RIS are located at $(0,0)$ and $(10,0)$ (in meters), respectively, 
whereas the positions of the UEs are generated as random variables
uniformly distributed  within a circular cluster centered in $(40, -10)$ (in meters), whose 
radius is set to $10$ m in the case of homogeneous users and 
$100$ m in the case of heterogeneous users (see Subsection~\ref{sec:dis-rand}).
All the ensemble averages  
(with respect to the relevant random channel parameters,
spatial positions of the UEs, and, wherever applicable,  the randomized reflection coefficients
of the RIS) are evaluated through $200$ independent runs.
For the main system parameters, we refer to the exemplifying scenario reported in 
Table~\ref{tab:example-1}.
The remaining parameters of the simulation setting 
are reported in Table~\ref{tab:example-2}.
\begin{table*}
\scriptsize
\centering
\caption{Parameters of the simulation setup.}
\label{tab:example-2}
\begin{tabular}{ccc}
\hline
\noalign{\vskip\doublerulesep}
\textbf{Parameters} & \textbf{Symbols} & 
\textbf{Values} 
\tabularnewline[\doublerulesep]
\hline
Effective isotropic radiated power & EIRP & $33$ dBm
\tabularnewline
Noise power at the UEs & $\pot_\text{noise}$ & $-100$ dBm
\tabularnewline
Inter-element spacing at the RIS & $d_\text{RIS}$ & $\lambda_0/2$
\tabularnewline
Azimuth and elevation angles 
at the RIS &  $\theta_\text{RIS}$ and  $\phi_\text{RIS}$ & Uniformly distributed
\tabularnewline
Alphabet of the reflection coefficients of the RIS & $\mathcal{R}$ & $\{\pm 1, \pm j\}$
\tabularnewline
Ricean factor for the transmitter-to-RIS channel & $\kappa$ & $3$
\tabularnewline
Path-loss exponent & $\eta$ & 1.6
\tabularnewline
\hline
\end{tabular}
\end{table*}
The channel variances $\sigma_g^2$ (from the transmitter to the RIS),
$\sigma_{h_k}^2$ (from the transmitter to the $k$-th UE),
and $\sigma_{f_k}^2$ (from the RIS to the $k$-th UE)
are chosen as 
$\sigma^2_\alpha = G_\alpha \, d_\alpha^{-\eta} \, \lambda_0^2/(4 \,\pi)^2$, 
for $\alpha \in \{g, h_k, f_k \}$, where 
$G_\alpha=4 \pi A_\text{RIS}/\lambda_0^2$ for 
the RIS \cite{Ell.2021}, with  $A_\text{RIS}=Q \, (\lambda_0/2)^2$ denoting 
the physical area of the RIS,  and $G_\alpha=5$ dBi for the UEs, 
$d_\alpha$ is the distance of the considered link, and $\eta$ is given in Table~\ref{tab:example-2}.
Regarding the training overhead for CSI acquisition, we set 
$N_\text{up-full}=K \, (Q+1)$ and $N_\text{down-full}=1$
in the case of full CSIT (see Fig.~\ref{fig:fig_2}), 
whereas $N_\text{par}=2$ in the case of partial CSIT 
(see Fig.~\ref{fig:fig_3}).

In all the forthcoming examples, we compare the 
(statistical) average performance of the following RIS designs:
(i) the sum-rate time-averaged capacity \eqref{eq:C-sum}, 
referred to as ``STV opt", which 
has been derived by assuming full CSIT 
and by using Algorithm~\ref{table:tab_1} for 
the optimization of the reflection coefficients of
the slowly time-varying RIS over each channel coherence time interval;
(ii) the sum-rate time-averaged capacity \eqref{eq:Capa-ave-rand}, 
referred to as ``RTV rand", which 
has been derived by assuming partial CSIT 
and randomized slot-by-slot variations in both space and time of 
the reflection coefficients of the rapidly time-varying RIS;
(iii) the sum-rate time-averaged capacity of 
the solution derived in Subsection~\ref{sec:PFS-fullCSIT},
referred to as ``RTV opt w/ PFS", which 
has been derived by assuming full CSIT 
and PFS, where the reflection coefficients 
of the rapidly time-varying RIS are optimized for each
time slot by using Algorithm~\ref{table:tab_1};
(iv) the sum-rate time-averaged capacity 
of the solution  derived in Subsection~\ref{sec:PFS-partCSIT},
referred to as ``RTV rand w/ PFS", which 
has been derived by assuming partial CSIT 
and PFS, where the reflection coefficients 
of the rapidly time-varying RIS are randomized slot-by-slot 
in both space and time.
Furthermore, we plot the average performance of the
downlink in the absence of an RIS without
fairness constraints, referred to as 
``w/o RIS".\footnote{In the considered simulation setting, the performance of 
a downlink in the absence of an RIS with PFS is significantly worse than that 
of the scenario ``w/o RIS" and, thus, we decided not to report the corresponding curves here.
}
For all the scheduling algorithms, we also show the fairness index 
\eqref{eq:fair-index}. 

\begin{figure*}[!t]
\begin{minipage}[b]{9cm}
\centering
\includegraphics[width=\linewidth]{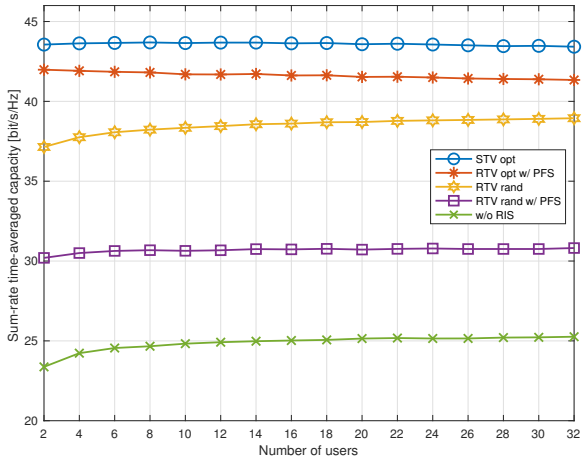}
\end{minipage}
\begin{minipage}[b]{9cm}
\centering
\includegraphics[width=\linewidth]{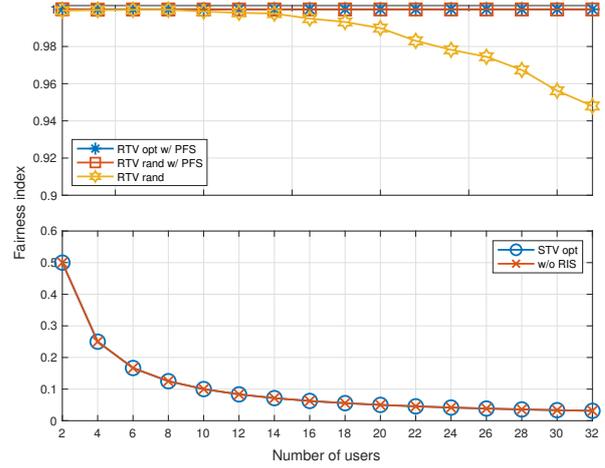}
\end{minipage}
\caption{Sum-rate time-averaged capacity (left) and fairness index (right) versus 
the number of users $K$ (Example~1).
}
\label{fig:fig_4}
\end{figure*}
\begin{figure*}[!t]
\begin{minipage}[b]{9cm}
\centering
\includegraphics[width=\linewidth]{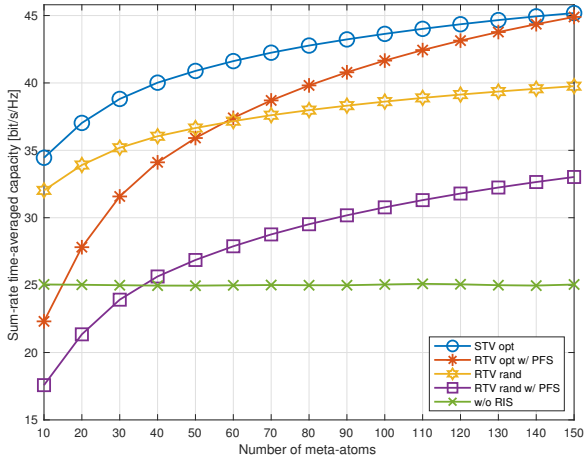}
\end{minipage}
\begin{minipage}[b]{9cm}
\centering
\includegraphics[width=\linewidth]{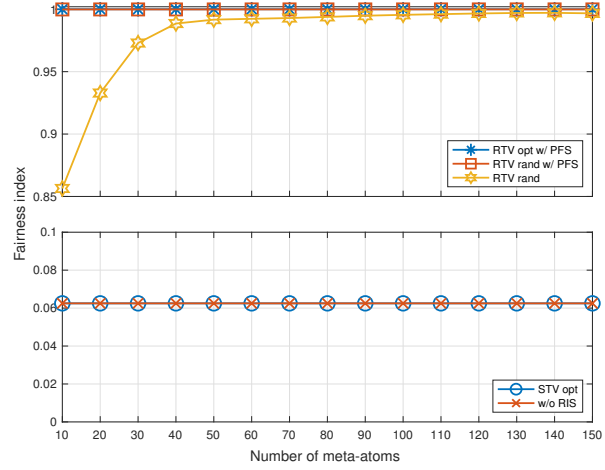}
\end{minipage}
\caption{Sum-rate time-averaged capacity (left) and fairness index (right) versus 
the number of meta-atoms $Q$ (Example~1).
}
\label{fig:fig_5}
\end{figure*}

\subsection{Example~1: Homogeneous users}

Fig.~\ref{fig:fig_4} shows the sum-rate time-averaged capacity 
(left-side plot) and fairness (right-side plot) performance
of the schemes under comparison in the case of 
homogeneous users as a function of the number 
of users $K$, with $Q=10 \times 10$ meta-atoms.
We observe that, except for the schemes requiring full CSIT,   
the sum-rate time-averaged 
capacity increases double logarithmically with the number
of users $K$ in the system: this is
the multiuser diversity effect \cite{Tse-book}.
The same effect is not visible for the  
``STV opt" and ``RTV opt w/ PFS" downlinks 
since the logarithmic growth of the 
mean overall channel with $K$ is contrasted 
by the pilot overhead required to acquire full CSIT, which
also depends on $K$.  
It is clear that the presence of an RIS leads to
significant sum-rate gains with respect the 
``w/o RIS" case. 

As expected, even thought the ``STV opt" downlink 
allows to achieve the best sum-rate capacity performance, it requires
full CSIT and, moreover, as in the scenario of ``w/o RIS", 
it becomes fairer and fairer
as the number of users increases. 
In contrast, from the curves of ``RTV rand", one can argue 
that, compared to  ``STV opt",  a much fairer system 
can be obtained by introducing a rapidly time-varying
randomization of the reflection response of the RIS, which also allows
to highlight an interesting trade-off between performance and 
required  amount of CSIT for downlink optimization. 
A $100 \%$ fair system can be designed by using PFS. Obviously,
such a fairness comes at the price of a sum-rate reduction:   indeed,
the sum-rate curves of the 
``RTV opt w/ PFS" and ``RTV rand w/ PFS"
systems are worse than those of their corresponding counterparts 
``STV opt" and ``RTV rand", respectively. 
Interestingly, although it involves partial CSIT only, 
the ``RTV rand" downlink pays
a mild performance gain with respect to 
``RTV opt w/ PFS" - which however requires full CSIT - 
by ensuring a fairness of nearly $100\%$ up to $18$ users 
and no smaller than $94 \%$ up to $32$ users.

Fig.~\ref{fig:fig_5} shows the sum-rate time-averaged capacity 
(left-side plot) and fairness (right-side plot) performance
of the schemes under comparison in the case of 
homogeneous users as a function of the number 
of meta-atoms $Q$, with $K=16$ users.
Two additional interesting trends can be inferred from such figures.
First, in accordance with our analysis in 
Subsection~\ref{sec:dis-rand}, 
the sum-rate time-averaged 
capacity of the schemes employing partial CSIT
increases logarithmically with the number
of meta-atoms $Q$ at the RIS, whereas the 
downlinks optimized with full CSIT exhibit 
a faster growth rate. 
Second, as the number of meta-atoms increases, 
the sum-rate capacity of ``RTV opt w/ PFS"
tends to that of ``STV opt" by ensuring at the same time 
$100 \%$ fairness, while $\fair =100/16=6.25 \%$ for the ``STV opt" scheme.

\begin{figure*}[!t]
\begin{minipage}[b]{9cm}
\centering
\includegraphics[width=\linewidth]{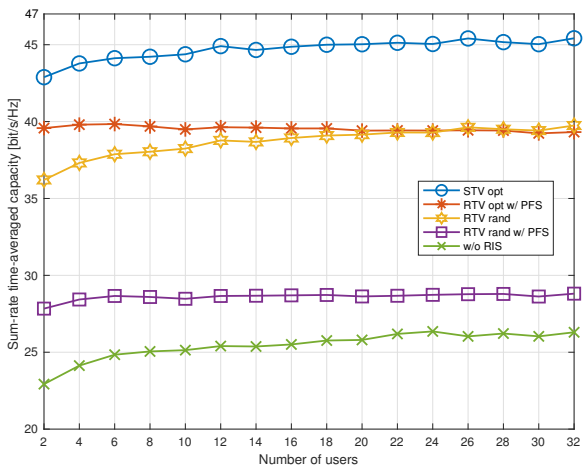}
\end{minipage}
\begin{minipage}[b]{9cm}
\centering
\includegraphics[width=\linewidth]{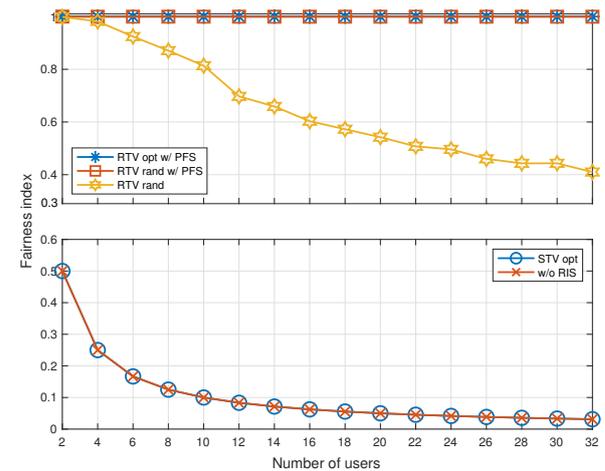}
\end{minipage}
\caption{Sum-rate time-averaged capacity (left) and fairness index (right) versus 
the number of users $K$ (Example~2).
}
\label{fig:fig_6}
\end{figure*}
\begin{figure*}[!t]
\begin{minipage}[b]{9cm}
\centering
\includegraphics[width=\linewidth]{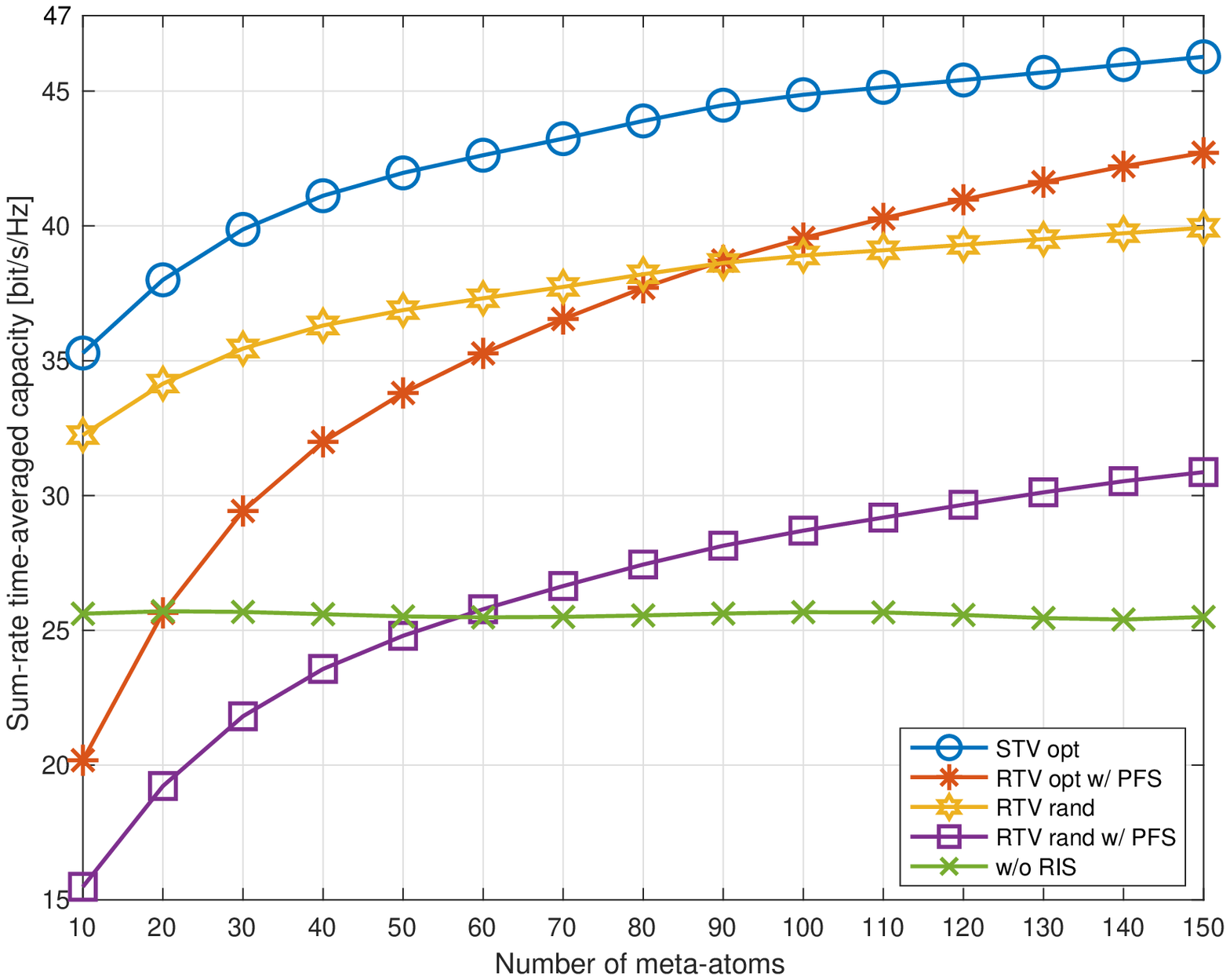}
\end{minipage}
\begin{minipage}[b]{9cm}
\centering
\includegraphics[width=\linewidth]{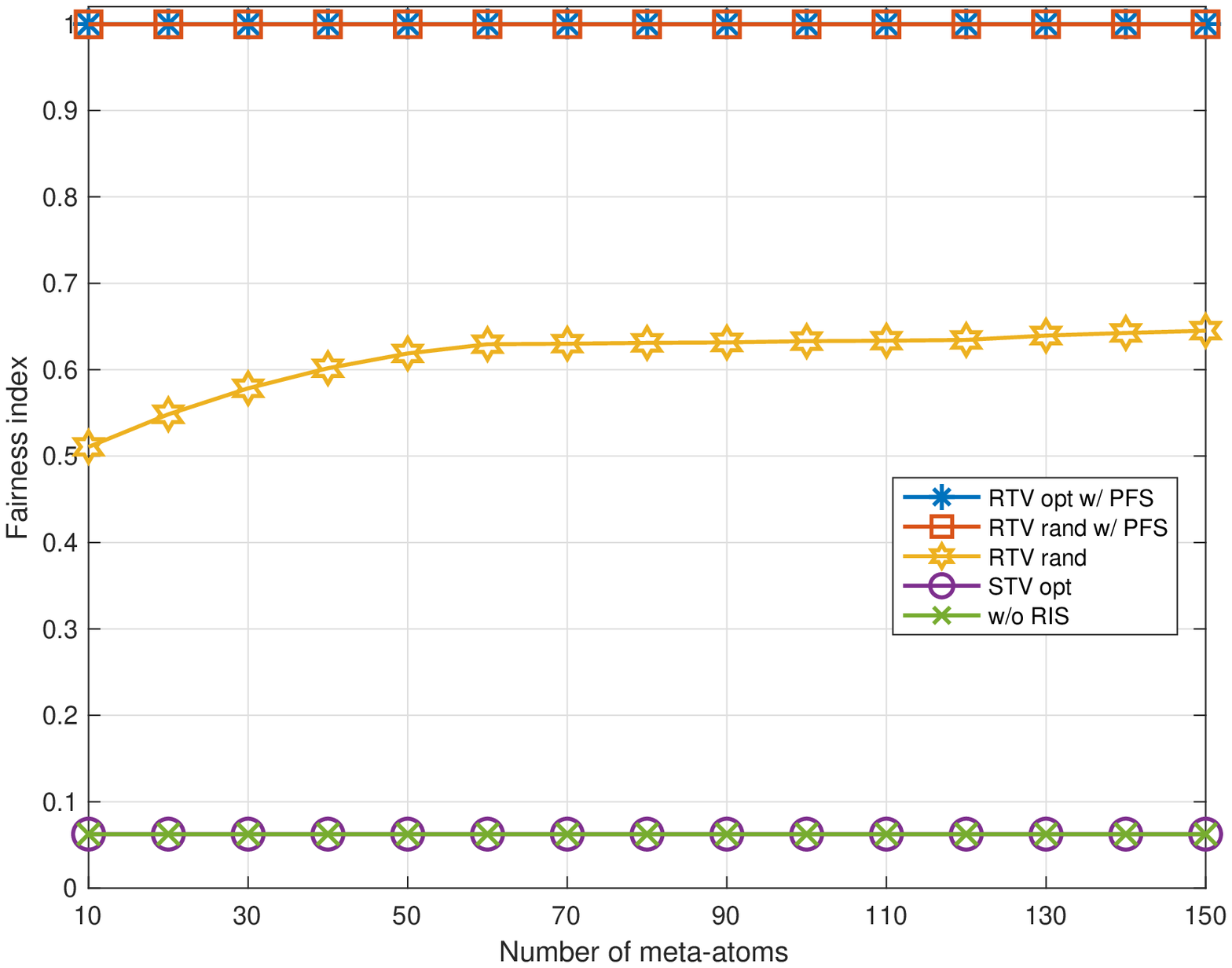}
\end{minipage}
\caption{Sum-rate time-averaged capacity (left) and fairness index (right) versus 
the number of meta-atoms $Q$ (Example~2).
}
\label{fig:fig_7}
\end{figure*}

\subsection{Example~2: Heterogeneous users}

Figs.~\ref{fig:fig_6} and \ref{fig:fig_7}
report the sum-rate time-averaged capacity 
(left-side plot) and fairness (right-side plot) performance
of the schemes under comparison in the case of 
heterogeneous users. In Fig.~\ref{fig:fig_6}, we set
$Q=10 \times 10$ meta-atoms, whereas  
we fix $K=16$ users in Fig.~\ref{fig:fig_7}.
Besides substantially confirming the above trends 
evidenced in Example~1, a behavior that is worth 
mentioning is the reduction in fairness of the  
``RTV rand" downlink, compared to the 
homogeneous case. Hence, we might argue that 
the randomization in both space and time of the
reflection coefficients of the RIS leads to a fair
resource allocation if there is no severe 
``near-far effect", i.e., users that experience 
significantly different fading effects.
It is apparent that, to cope with
such an heterogeneous situation,  the schemes employing PFS
are forced to meaningfully cut down their sum-rate performance 
in order to fulfill the fairness constraint.

\section{Conclusion and directions for future work}
\label{sec:conclusions}

In this article, we have investigated the impact of 
a rapidly time-varying RIS on the resource allocation in 
a downlink multiuser system operating over 
slow fading channels. 

The following results have been obtained:

\begin{enumerate}

\item 
Joint optimization of transmit powers and reflection coefficients
of the RIS in the case of full CSIT leads to a slowly time-varying RIS,
whose reflection response is updated with a rate that is 
(approximately) the inverse of the channel coherence time. 
In this case, only one user is scheduled within  a channel
coherence interval.

\item 
The assumption of full CSIT can be relaxed by randomizing 
the reflection coefficients of a rapidly time-varying RIS in both 
space and time domains, without paying a penalty 
in terms of multiuser diversity gain.  
The resulting resource allocation 
scheme requires partial CSIT only and 
ensures a certain fairness, since the randomization in time allows 
the transmitter to serve 
more than one user for each channel coherence time interval.

\item 
If the main design concern is to 
meet fairness among users while at the same time exploiting
the multiuser diversity gain, then a rapidly time-varying RIS 
has to be employed in the downlink.
In such a case, the temporal dimension of the RIS is 
beneficial not only in the partial CSIT scenario, but also
in the full CSIT one.
 
\end{enumerate}

This study demonstrates that, if acquisition of 
full CSIT is not possible in practice and/or the communication
system has stringent requirements in terms of fairness, the 
exploitation of the temporal dimension of a rapidly time-varying 
RIS is instrumental in designing the downlink of a multiuser systems
even in slow fading scenarios.
We have considered single-antennas transmitter and receivers.
In this respect, a first interesting research subject consists of extending
the proposed framework to the case of multi-antenna terminals.
In this case, the downlink is in general
a nondegraded broadcast channel and 
sum capacity is achieved by using more sophisticated 
multiple access strategies than opportunistic time sharing, i.e., 
dirty-paper coding.
An additional research direction involves investigating
the benefits of a rapidly time-varying RIS in fast fading environments.


\begin{IEEEbiography}[
{\includegraphics[width=1in,height=1.25in,clip,keepaspectratio]{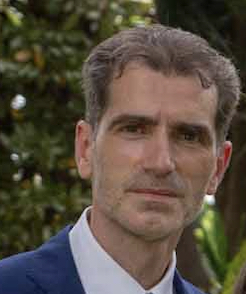}}]
{Francesco Verde}(Senior Member, IEEE)  received the Dr. Eng. degree
\textit{summa cum laude} in electronic engineering
from the Second University of Napoli, Italy, in 1998, and the Ph.D.
degree in information engineering
from the University of Napoli Federico II, in 2002.
Since December 2002, he has been with the University of Napoli Federico II, Italy. He first served as an Assistant Professor of signal theory and mobile communications
and, since December 2011, he has served as an Associate Professor of telecommunications with the Department of Electrical Engineering and Information Technology.
His research activities include reflected-power communications, 
orthogonal/non-orthogonal multiple-access techniques, wireless systems optimization, and 
physical-layer security.

Prof. Verde has been involved in several technical program committees of major IEEE conferences in signal processing and wireless communications.
He has served as Associate Editor for IEEE TRAN\-SACTIONS ON VEHICULAR TECHNOLOGY since 2022.
He was an Associate Editor of the IEEE TRANSACTIONS ON SIGNAL PROCESSING (from 2010 to 2014), IEEE SIGNAL PROCESSING LETTERS (from 2014 to 2018),
IEEE TRANSACTIONS ON COMMUNICATIONS (from 2017 to 2022), and 
Senior Area Editor of the IEEE SIGNAL PROCESSING LETTERS (from 2018 to 2023), 
as well as Guest Editor of the EURASIP Journal on Advances in Signal Processing in 2010 and SENSORS MDPI in 2018-2022.
\end{IEEEbiography}

\vspace*{-5\baselineskip}

\begin{IEEEbiography}[
{\includegraphics[width=1in,height=1.25in,clip,keepaspectratio]{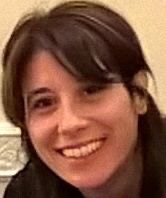}}]
{Donatella Darsena} (Senior Member, IEEE) received the Dr. Eng. degree summa cum laude in telecommunications engineering in 2001, and the Ph.D. degree in electronic and telecommunications engineering in 2005, both from the University of Napoli Federico II, Italy. From 2001 to 2002, she worked as embedded system designer in the Telecommunications, Peripherals and Automotive Group, STMicroelectronics, Milano, Italy. 
In 2005 she joined the Department of Engineering at Parthenope University of Napoli, Italy and worked first as an Assistant Professor and then as an Associate Professor from 2005 to 2022.
She is currently an Associate Professor in the Department of Electrical Engineering and Information Technology of the University of Napoli Federico II, Italy.
Her research interests are in the broad area of signal processing for communications, with current emphasis on reflected-power communications, orthogonal and nonorthogonal multiple access techniques, wireless system optimization, and physical-layer security.
Dr. Darsena has served as an Associate Editor for IEEE ACCESS since October 2018 and for IEEE SIGNAL PROCESSING LETTERS since 2020. 
She served first as an Associate Editor (from December 2016 to July 2019) and then as a Senior Area Editor (from August 2019 to July 2023) for IEEE COMMUNICATIONS LETTERS.
Since July 2023, she has been an Executive Editor of IEEE COMMUNICATIONS LETTERS.
\end{IEEEbiography}

\vspace*{-5\baselineskip}

\begin{IEEEbiography}[{\includegraphics[width=1in,height=1.25in,clip,keepaspectratio]{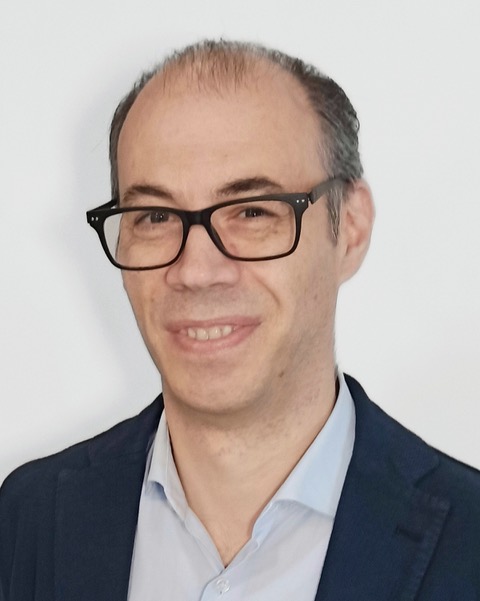}}]
	{Vincenzo Galdi} (Fellow, IEEE) 
	received the Laurea degree ({\em summa cum laude}) in electrical engineering and the Ph.D. degree in applied electromagnetics from the University of Salerno, Italy, in 1995 and 1999, respectively. 
	
	He has held several research-associate and visiting positions at abroad research institutions, including the European Space Research and Technology Centre, Noordwijk, The Netherlands; Boston University, Boston, MA, USA; the Massachusetts Institute of Technology, Cambridge, MA, USA; the California Institute of Technology, Pasadena, CA, USA; and The University of Texas at Austin, Austin, TX, USA. He is currently a Professor of electromagnetics with the Department of Engineering, University of Sannio, Benevento, Italy, where he leads the Fields \& Waves Laboratory. He is the Co-Founder of the spinoff company MANTID srl, Benevento, and the startup company BioTag srl, Naples. He has co-edited two books and coauthored about 180 articles in peer-reviewed international journals, and is the co-inventor of thirteen patents. His research interests encompass wave interactions with complex structures and media, multiphysics metamaterials, smart propagation environments, optical sensing, and gravitational interferometry. 
	
	Dr. Galdi is a Fellow of Optica (formerly OSA), a Senior Member of the LIGO Scientific Collaboration, and a member of the American Physical Society. He was a recipient of the Outstanding Associate Editor Award of \textsc{IEEE Transactions on Antennas and Propagation} in 2014 and the URSI Young Scientist Award in 2001. He has served as the Chair for the Technical Program Committee of the International Congress on Engineered Material Platforms for Novel Wave Phenomena in 2018, the Topical/Track Chair for the Technical Program Committee of the IEEE International Symposium on Antennas and Propagation and USNC-URSI Radio Science Meeting from 2016 to 2017 and from 2020 to 2023, and an organizer/chair for several topical workshops and special sessions. He has also served as a Track Editor from 2016 to 2020, a Senior Associate Editor from 2015 to 2016, and an Associate Editor from 2013 to 2014 of the \textsc{IEEE Transactions on Antennas and Propagation}. He is serving as an Associate Editor of {\em Optics Express} and {\em Heliyon} and a regular reviewer for many journals, conferences, and funding agencies.
	
\end{IEEEbiography}

\end{document}